\documentclass[11pt]{article}

\usepackage{dsfont}

\def\colorful{1}

\usepackage{amsthm,amsfonts,amsmath,amssymb,epsfig,color,float,graphicx,verbatim, enumitem}
\usepackage{multirow}
\usepackage[noend]{algpseudocode}
\usepackage{enumitem}
\usepackage[linesnumbered,lined,boxed,commentsnumbered,ruled,vlined]{algorithm2e}

\newif\ifhyper\IfFileExists{hyperref.sty}{\hypertrue}{\hyperfalse}
\hypertrue
\ifhyper\usepackage{hyperref}\fi

\usepackage{enumitem}

\usepackage{framed}
\usepackage{nicefrac}

\def\nnewcolor{0}
\ifnum\nnewcolor=1
\newcommand{\nnew}[1]{{\color{blue} #1}}
\fi
\ifnum\nnewcolor=0
\newcommand{\nnew}[1]{#1}
\fi

\ifnum\colorful=1

\else

\fi

\newtheorem{theorem}{Theorem}[section]

\newtheorem{lemma}[theorem]{Lemma}
\newtheorem{informal theorem}[theorem]{Theorem (informal statement)}

\newtheorem{proposition}[theorem]{Proposition}
\newtheorem{corollary}[theorem]{Corollary}
\newtheorem{claim}[theorem]{Claim}
\newtheorem{fact}[theorem]{Fact}

\newtheorem{remark}[theorem]{Remark}

\theoremstyle{definition}
\newtheorem{definition}[theorem]{Definition}
\newcommand{\eqdef}{\stackrel{{\mathrm {\footnotesize def}}}{=}}

\newcommand{\size}{\mathrm{size}}

\newcommand{\R}{\mathbb{R}}

\newcommand{\Z}{\mathbb{Z}}
\newcommand{\N}{\mathbb{N}}
\newcommand{\E}{\mathbf{E}}

\newcommand{\eps}{\epsilon}
\newcommand{\dtv}{d_{\mathrm{TV}}}

\renewcommand{\Pr}{\mathbf{Pr}}
\newcommand{\poly}{\mathrm{poly}}
\newcommand{\var}{\mathbf{Var}}

\newcommand{\Cov}{\mathbf{Cov}}
\newcommand{\polylog}{\mathrm{polylog}}

\newcommand{\D}{\mathcal{D}}

\newcommand{\littlesum}{\mathop{\textstyle \sum}}

\newcommand{\wt}{\widetilde}
\newcommand{\wh}{\widehat}

\newcommand{\proj}{\mathrm{proj}}

\newcommand{\relu}{\mathrm{ReLU}}

\title{Implicit High-Order Moment Tensor Estimation\\
and Learning Latent Variable Models}

\author{
Ilias Diakonikolas\thanks{Supported by NSF Medium Award CCF-2107079, and an H.I. Romnes Faculty Fellowship.}\\
University of Wisconsin-Madison\\
{\tt ilias@cs.wisc.edu}\\
\and
Daniel M. Kane\thanks{Supported by NSF Medium Award CCF-2107547 and NSF Award CCF-1553288 (CAREER).}\\
University of California, San Diego\\
{\tt dakane@cs.ucsd.edu}
}

\date{}

\begin{document}

\maketitle

\vspace{-0.6cm}

\begin{abstract}
We study the general task of learning latent-variable models on $\mathbb{R}^d$
with $k$ hidden parameters.
A common technique to address this task algorithmically 
is (some version of) the method of moments.
Unfortunately, moment-based approaches are often hampered by the fact
that the moment tensors of super-constant degree
cannot even be written down in polynomial time.
Motivated by such learning applications, we develop a general efficient algorithm
for {\em implicit moment tensor computation}. Roughly speaking, our algorithm
computes in $\mathrm{poly}(d, k)$ time a succinct approximate description of tensors
of the form $M_m = \sum_{i=1}^{k} w_i v_i^{\otimes m}$,
for $w_i \in \mathbb{R}_+$---even for $m=\omega(1)$---assuming that there exists
an unbiased estimator for $M_m$ with small variance that takes an appropriately nice form. 
Our framework broadly generalizes, both conceptually and technically,
the work of~\cite{LL21-opt} which developed an efficient algorithm for the specific moment tensors
that arise in the task of clustering mixtures of spherical Gaussians.

By leveraging our implicit moment estimation algorithm, we obtain the first
$\mathrm{poly}(d, k)$-time learning algorithms
for the following classical latent-variable models---thereby resolving
or making significant progress towards a number of important open problems in the literature.
\begin{itemize}[leftmargin=*]
\item {\bf Mixtures of Linear Regressions}
Given i.i.d.\ samples $(x, y)$ with $x \sim N(0, I)$ and such that the joint distribution on $(x, y)$
 is an unknown $k$-mixture of linear regressions on $\mathbb{R}^{d+1}$ corrupted with Gaussian noise,
the goal is to learn the underlying distribution in total variation distance.
We give a $\mathrm{poly}(d, k, 1/\epsilon)$-time algorithm for this task, where $\epsilon$ is the desired error.
The previously best algorithm has super-polynomial complexity in $k$.

\item {\bf Mixtures of Spherical Gaussians} Given i.i.d.\ samples
from a $k$-mixture of identity covariance Gaussians on $\mathbb{R}^d$,
the goal is to learn the target mixture. For density estimation, 
we give a $\mathrm{poly}(d, k, 1/\epsilon)$-time learning algorithm,
where $\epsilon$ is the desired total variation error,
under the condition that the means lie in a ball of radius $O(\sqrt{\log k})$.
Prior algorithms incur super-polynomial complexity in $k$.
For parameter estimation, we give a $\mathrm{poly}(d, k, 1/\epsilon)$-time algorithm
where $\epsilon$ is the target accuracy,
under the {\em optimal} mean separation of $\Omega(\log^{1/2}(k/\epsilon))$
and the condition that the largest
distance is comparable to the smallest.
Prior polynomial-time parameter estimation algorithms require separation
$\Omega(\log^{1/2+c}(k/\epsilon))$, for  $c>0$.

\item {\bf Positive Linear Combinations of Non-Linear Activations}
Given i.i.d.\ samples $(x, y)$ with
$x \sim N(0, I)$ and $y = F(x)$, where $F$ is a positive linear combination
of $k$ reasonable non-linear activations on $\mathbb{R}^d$,
the goal is to learn the target function in $L_2$-norm. 
Our main result is a general algorithm for this task with complexity $\poly(d, k) g(\epsilon)$, 
where $\epsilon$ is the desired error and the function $g$ depends 
on the Hermite concentration of the target class of functions. 
Specifically, for positive linear combinations of ReLU activations, our algorithm has complexity
$\mathrm{poly}(d, k) 2^{\mathrm{poly}(1/\epsilon)}$. This is the first algorithm for this class  
that runs in $\poly(d, k)$ time for {\em sub-constant} values of $\epsilon = o_{k, d}(1)$.
Finally, for positive linear combinations of cosine activations with bounded frequency, 
our algorithm runs in $\mathrm{poly}(d, k,1/\epsilon)$ time. 
\end{itemize}
\end{abstract}


\setcounter{page}{0}

\thispagestyle{empty}

\newpage

\section{Introduction} \label{sec:intro}

This work is motivated by the general algorithmic problem of learning
probabilistic models with latent variables in high dimensions.
This is a prototypical statistical estimation task,
first studied in the pioneering work of Karl Pearson from 1894~\cite{Pearson:94}.
Pearson introduced the method of moments motivated by the problem of
learning Gaussian mixtures in one dimension---one of the most basic latent-variable models.
The algorithmic version of learning high-dimensional latent variable models,
including mixtures of Gaussians, has been extensively studied in recent decades
within the theoretical computer science and machine learning communities.
The relevant literature is vast and has resulted in significant algorithmic
advances for a diverse range of settings.

Here we focus on the regime where both the dimension $d$ and
the number of hidden parameters $k$ are large.
In this context, we are interested in designing
learning algorithms for natural models with complexity $\poly(d, k)$, i.e.,
polynomial in both $d$ and $k$. Of course,
there are examples where such a complexity upper bound may not be possible,
due to fundamental computational limitations.
A prominent such example (where $\poly(d, k)$ complexity is unlikely) 
is for the task of learning (proper learning or density estimation)
$k$-mixtures of Gaussians on $\R^d$ with general---aka non-spherical---covariances.
For this learning task, recent work~\cite{DKS17-sq, Bruna0ST21, GupteVV22, DiaKPP24-sos}
has provided strong evidence of computational hardness---ruling out $\poly(d, k)$ runtime,
even though $\poly(d, k)$ samples information-theoretically suffice.

For the latent-variable models we study in this paper---namely
$k$-mixtures of linear regressions, $k$-mixtures of spherical Gaussians, and
one-hidden-layer neural networks with $k$ gates---
there is no known formal evidence of hardness.
On the other hand, prior to this work,
the best known algorithms
had {\em super-polynomial} complexity (as a function of the parameter $k$).
{\em By leveraging our main result (Proposition~\ref{prop:main-informal}),
we design the first $\poly(d, k)$ time learning algorithms
for all of these problems. We thus resolve, or make significant progress towards,
several well-known open problems in the learning theory literature.}
All of our learning algorithms rely on a novel methodology---which we
term {\em implicit moment tensor estimation}---that is
of broader interest and we believe may find other applications.

Before we describe our main technique, we start by highlighting
its algorithmic implications on learning latent-variable models.

\subsection{Efficiently Learning Latent Variable Models} \label{ssec:intro-apps}

\paragraph{Mixtures of Linear Regressions}
A $k$-mixture of linear regressions ($k$-MLR), specified by
mixing weights $w_i \geq 0$, where $\sum_{i=1}^k w_i=1$,
and regressors $\beta_i \in \R^d$, $i \in [k]$, is
the distribution $F$ on pairs $(x, y) \in \R^d \times \R$, where $x \sim N(0, I)$
and with probability $w_i$ we have  that $y = \beta_i \cdot x + \nu$, where
$\nu \sim N(0, \sigma^2)$ is independent of $x$.
Mixtures of linear regressions are a classical probabilistic
model introduced in~\cite{DV89, JJ94}
and have since been extensively studied in machine learning;
see Section~\ref{ssec:related} for a detailed summary of prior work.

Here we focus on density estimation for $k$-MLRs,
where the goal is to learn the underlying distribution in total variation distance---without making
any separation assumptions on the pairwise distance between the components.
While density estimation can be information-theoretically solved with $\poly(d, k)$ samples,
all prior algorithms required super-polynomial time. Specifically, \cite{LL18} gave an algorithm
with complexity exponential in $k$. This bound was improved by \cite{CLS20}
to sub-exponential, namely scaling with $\exp(\tilde{O}(k^{1/2}))$.
The fastest previously known algorithm~\cite{DK20-ag} has
quasi-polynomial complexity in $k$, namely $\poly(d, (k/\eps)^{\log^2(k)})$,
where $\eps$ is the accuracy.

As our main contribution for this problem, we give the first $\poly(d, k)$ time algorithm.
Specifically, we show (see Theorem~\ref{MLRThm}):

\begin{theorem}[Density Estimation for $k$-MLR]\label{MLRThm-inf}
Let $F$ be a $k$-MLR distribution on $\R^{d+1}$ with $B,\sigma \leq 1$,
where $B = \max_i \|\beta_i\|_2$ and $\sigma>0$ is the standard deviation of the Gaussian noise.
There exists an algorithm that
draws $N=\poly(k,d)(1/\eps)^{O(\sigma^{-2})}$ samples from $F$,
runs in $\poly(N, d)$ time, and returns a sampler for a distribution that is $\eps$-close to $F$ in total variation distance.
\end{theorem}

We note that the existence of a learning algorithm for $k$-MLR
with a $\poly(d, k)$ complexity had been a well-known open problem
in the literature. This question was most recently posed as one of the main open problems in~\cite{CLS20};
see the associated talk at STOC 2020~\cite{Chen20-youtube}.
Theorem~\ref{MLRThm-inf} answers this open question in the affirmative.

Some additional remarks are in order.
For the setting that the standard deviation $\sigma$ is at least a positive
universal constant, the complexity of our density estimation algorithm is $\poly(k, d, 1/\eps)$---essentially resolving
the complexity of this learning task. If $\sigma$ is very small (but still non-zero),
for example for $\sigma = \Theta(\eps)$ (a regime considered in \cite{CLS20}),
the complexity of our learning algorithm becomes
$\poly(k, d) 2^{\tilde{O}(1/\eps^2)}$.  Even in this regime, we obtain the
first learning algorithm for $k$-MLR that runs in $\poly(d, k)$ time
for {\em sub-constant} values of $\eps = o(1)$---namely up to
$\eps = 1/\log^{1/2}(dk)$. The fastest previous algorithm~\cite{DK20-ag}
has complexity {\em quasi-polynomial} in $k$ even for constant $\eps$.
An interesting question left open by our work is to
remove the dependence on $\sigma$ from the complexity of our algorithm.

\paragraph{Mixtures of Spherical Gaussians}
A $k$-mixture of spherical Gaussians is any distribution on $\R^d$
of the form $F = \sum_{i=1}^k w_i N(\mu_i,  I)$,
where $\mu_i \in \R^d$ are the unknown mean vectors and $w_i \geq 0$,
with $\sum_{i=1}^k w_i = 1$, are the mixing weights.
Gaussian mixture models are one of the oldest and most extensively
studied latent variable models~\cite{Pearson:94}; see Section~\ref{ssec:related}
for the most relevant prior work.

We study both density estimation and parameter estimation for spherical $k$-GMMs.
In density estimation, the goal is to output a hypothesis
within small total variation distance from the target. In parameter estimation,
the goal is to learn the mean vectors to any desired accuracy---under some necessary
assumptions on the pairwise distance between the component means.

Prior to our work, the fastest density estimation algorithm~\cite{DK20-ag}
had sample and computational complexity $\poly(d) (k/\eps)^{O(\log^2(k))}$.
Here we achieve a polynomial time algorithm
under a boundedness assumption on the component means
(see Theorem~\ref{thm:formal-dens-GMMs}):

\begin{theorem}[Density Estimation for Spherical $k$-GMMs with Bounded Means]\label{thm:inf-dens-GMMs}
There is an algorithm that given $\eps>0$ and $n=\poly(d, k, 1/\eps)$ samples
from a $k$-mixture of spherical Gaussians $F$ on $\R^d$ with component means in a ball of radius
$O(\sqrt{\log(k)})$, it runs in $\poly(n, d)$ time and outputs a
hypothesis distribution $H$ such that with high probability $\dtv(H, F) \leq \eps$.
\end{theorem}

\noindent This is the first polynomial-time algorithm for density estimation
of a natural subclass of spherical mixtures. Theorem~\ref{thm:inf-dens-GMMs} makes significant progress
towards the resolution of a longstanding open problem in the  literature, most recently posed
in the STOC 2022 talk~\cite{Liu22-youtube} associated with~\cite{LL21-opt}.

We now switch our attention to parameter estimation.
The state-of-the-art result prior to our work is due to~\cite{LL21-opt}.
For the case of uniform mixtures\footnote{For simplicity of exposition, we focus on the uniform case here.
Both the results of \cite{LL21-opt} and ours hold for arbitrary weights with modified guarantees.},
they
gave a $\poly(d, k)$ time algorithm under the assumption that the minimum
pairwise mean separation $s = \min_{i\neq j} \|\mu_i-\mu_j\|_2$ satisfies
$s \gg (\log k)^{1/2+c}$ for any small constant $c>0$.
The information-theoretically optimal separation under which this task
is solvable with $\poly(d, k)$ samples is $s \gg (\log k)^{1/2}$~\cite{RV17-mixtures}.
We give a $\poly(d, k, 1/\eps)$ time parameter estimation algorithm
under the optimal separation, with the additional assumption that the
largest pairwise distance is comparable to the smallest.
Specifically, we establish the following (see Theorem~\ref{thm:formal-cluster-GMMs}):

\begin{theorem}[Clustering Mixtures of Spherical Gaussians under Optimal Separation] \label{thm:inf-cluster-GMMs}
Let $F = \sum_{i=1}^k (1/k) N(\mu_i,I)$
and $\eps > 0$.
Suppose that the minimum separation $s := \min_{i\neq j}\|\mu_i - \mu_j\|_2$
is at least a sufficiently large constant multiple of
$\sqrt{\log(k/\eps)}$ and furthermore that
$\max_{i\neq j} \|\mu_i - \mu_j\|_2 = O(\min_{i\neq j} \|\mu_i - \mu_j\|_2).$
There exists an algorithm that given
$n=\poly(d, k, 1/\eps)$ i.i.d.\ samples from $F$, runs in $\poly(n, d)$ time, and
outputs estimates $\tilde \mu_i$ such that with high probability,
for some permutation $\pi$ of $[k]$, it holds
$\|\tilde \mu_i - \mu_i\|_2 \leq \eps$, for all $i\in [k]$.
\end{theorem}

\noindent
Prior work~\cite{DKS18-list, HL18-sos, KSS18-sos, DK20-ag}
had given quasi-polynomial time parameter estimation algorithms
under the optimal separation. In the same context, the polynomial-time
algorithm of \cite{LL21-opt} succeeds under the stronger condition
that $s \gg \log(k/\eps)^{1/2+c}$, for a constant $c>0$.
We note that they do not require the extra assumption
that the largest pairwise distance is comparable to the smallest.
To remove this, they developed a complicated recursive clustering
argument (see Sections 10 and 11 of \cite{LL21-opt}),
which we expect can also be applied mutatis mutandis to our setting.

\paragraph{Learning One-hidden-layer Neural Networks}
A one-hidden-layer neural network
is any function $F: \R^d \to \R$ of the form
$F(x) = \sum_{i=1}^k w_i \sigma(v_i\cdot x)$ for some $w_i \in \R$
and unit vectors $v_i \in \R^d$, where $\sigma: \R \to \R$ is a known activation.
In this paper, we focus on the setting that the weights $w_i$ are positive
(for general weights, recent work has given evidence of computational hardness ruling out $\poly(d, k)$ runtime;
see Remark~\ref{rem:general-NNs}).

Let $\mathcal{C}_{\sigma, d, k}$ be the corresponding class of functions.
The PAC learning problem for $\mathcal{C}_{\sigma, d, k}$ is the following:
The input is a multiset of i.i.d.\ labeled examples $(x, y)$, where $x \sim N(0,I)$ and $y = F(x)$,
for some $F \in \mathcal{C}_{\sigma, d, k}$, where the activation $\sigma$ is known to the algorithm.
The goal is to output a hypothesis $H: \R^d \to \R$ that with high probability
is close to $F$ in $L_2$-norm.

A prominent special case corresponds to $\sigma(u) = \relu(u) \eqdef \max\{ 0, u\}$.
The learnability of one-hidden-layer ReLU networks has been extensively
studied over the past decade;
see, e.g.,~\cite{GeLM18, BakshiJW19, DKKZ20, DK20-ag, ChenKM21,  CDGKM23, ChenN24, DK23-schur}.

A central question in the foundations of deep learning is whether this class
is learnable in polynomial time,
when the marginal distribution is Gaussian. Quoting \cite{ChenKM21} 
(see also Chen's thesis~\cite{Chen-thesis}):
\begin{quote}
``Ideally, we would like an algorithm with sample complexity and running time
that is polynomial in all the relevant parameters. Even for learning arbitrary sums of ReLUs, [...],
it remains {\bf a major open question} to obtain a polynomial-time algorithm [...].''
\end{quote}

We give a general algorithm for $\mathcal{C}_{\sigma, d, k}$ that succeeds for a wide
class of activations $\sigma$, including ReLUs and cosine activations. 
Roughly speaking, the only properties required on $\sigma$ is that it has bounded fourth moment
and non-vanishing even-degree Hermite coefficients. Specifically, we prove the following
(see Theorem~\ref{thm:formal-sum-ReLUs}):

\begin{theorem}[PAC Learning $\mathcal{C}_{\sigma, d, k}$] \label{thm:inf-sum-ReLUs}
Suppose that the activation $\sigma$ has $\|\sigma\|_4 = O(1)$.
Letting $c_{\sigma,t} := \E[\sigma(G)h_t(G)]$ be the $t^{th}$ Hermite coefficient of $\sigma$,
suppose additionally that for some $\eps,\delta>0$ and some $n \in \Z_+$ it holds:
(i) $\sum_{t>n} c_{\sigma,t}^2 < \eps^2/4$, and (ii) $|c_{\sigma,t}| \geq \delta$
for all $1\leq t \leq 2n$ unless $t$ is odd and $c_{\sigma,t}=0$.

Then there is an algorithm that given $\eps$, $\sigma, k, d$, and access to examples $(x, F(x))$,
where $x \sim N(0, I)$ and $F \in \mathcal{C}_{\sigma, d, k}$, it runs in time
$\poly(dk/(\eps\delta)) 2^{O(n)}$, and outputs a hypothesis $H: \R^d \to \R$ such that with high probability
$\|H-F\|_2  \leq \eps$.
\end{theorem}

For the special case where $\sigma(u) = \relu(u)$, this leads to the following 
(see also Corollary~\ref{cor:sum-relu}):

\begin{corollary} \label{cor:sum-relu-intro}
There is a PAC learning algorithm for positive linear combinations of $k$ ReLUs on $\R^d$
with complexity $\poly(d,k) 2^{\poly(1/\eps)}$.
\end{corollary}

Corollary~\ref{cor:sum-relu-intro} makes significant progress
towards resolving the above open problem~\cite{ChenKM21, Chen-thesis}.
Specifically, it gives the first learning algorithm for sums of ReLUs
that runs in $\poly(d, k)$ time for {\em sub-constant} values of $\eps = o(1)$---namely up to
$\eps = 1/\log^c(dk)$, for some constant $c>0$.
The fastest previous algorithm~\cite{DK20-ag}
has sample and computational complexity $\poly(d) (k/\eps)^{O(\log^2(k))}$, i.e.,
quasi-polynomial in $k$ even for $\eps = \Theta(1)$.
On the other hand, the running time of the \cite{DK20-ag} algorithm is
better than ours for $\eps$ very small (e.g., $1/\poly(d, k)$).

Our second application is for positive linear combinations of periodic activations, 
in particular cosine activations. Such periodic activations are of particular interest in signal processing and computer vision, 
where they have been empirically observed to be more accurate that ReLU networks 
in representing natural signals; 
see, e.g.,~\cite{SitzmannMBLW20, MildenhallSTBRN20,RaiNMZYYFMPRABC21, VargasPHA24}. 


For $\sigma(u) = \cos(\gamma u)$, where $\gamma>0$ is a parameter, 
we obtain the following (see Corollary~\ref{cor:sum-cos}):

\begin{corollary} \label{cor:sum-cos-intro}
For $\sigma(u) = \cos(\gamma u)$, for some frequency parameter $\gamma > 0$, 
there is a PAC learning algorithm for $\mathcal{C}_{\sigma,d,k}$ 
with complexity $2^{O(1/\gamma^2)}\poly(dk/\eps)$.
\end{corollary}

\noindent While the dependence on $\gamma$ in our algorithm is exponential, 
this is inherent due to matching computational lower 
bounds~\cite{SongVW017, SongZB21, DiakonikolasKRS23} (even for $k=1$). 
Note that for $\gamma \ll 1/\sqrt{\log(dk/\eps)}$, our algorithm runs in 
$\poly(dk/\eps)$ time.

\begin{remark} \label{rem:general-NNs}
{\em For one-hidden-layer ReLU networks with {\em general} (i.e., positive or negative)  weights,
the fastest known algorithms take time $(d/\eps)^{\poly(k)}$~\cite{CN23, DK23-schur}.
Moreover, cryptographic hardness results~\cite{LZZ24} rule out
the existence of an algorithm with $\poly(d, k)$ runtime.}
\end{remark}


\subsection{Main Result: General Implicit Moment Tensor Estimation} \label{ssec:intro-main}

All aforementioned learning results are obtained as
applications of a general implicit moment tensor computation algorithm that we design,
which is the focus of this section.
We start by describing our framework, under which such an algorithm is possible,
followed by an informal statement of our main result (Proposition~\ref{prop:main-informal})
and a comparison with the most relevant prior work~\cite{LL21-opt}.

\vspace{-0.3cm}

\paragraph{Background on Tensors}
In order to proceed, we require basic background on tensors
(see Section~\ref{sec:prelims} for more details).
For a vector $x \in \R^d$, we denote by $x^{\otimes m}$ the $m$-th order tensor power of $x$.
Let $\otimes$ denote the tensor/Kronecker product between vector spaces.
Fix some $t \in \Z_+$. For vector spaces $V_i$, $i \in [t]$,
an order-$t$ tensor is an element $A \in \bigotimes_{i=1}^t V_i$.
If $A, B \in \bigotimes_{i=1}^t V_i$, for some inner product spaces $V_i$, $i \in [t]$,
then we use  $\langle A, B \rangle$ to denote the inner product of $A$
and $B$ {\em induced by the inner products on $V_i$}.
Namely, using the unique inner product structure on $\bigotimes_{i=1}^t V_i$
so that
$\langle v_1\otimes \cdots \otimes v_t, w_1\otimes \cdots \otimes w_t\rangle
= \prod_{i=1}^t \langle v_i, w_i\rangle$ for all vectors $v_i,w_i \in V_i$.
For the special case that $V_i = \R^{d_i}$, $d_i \in \Z_+$,
a tensor $A$ is defined by a $t$-dimensional array with real entries $A_{\alpha}$,
where $\alpha = (\alpha_1, \ldots, \alpha_t)$ with $\alpha_i \in [d_i]$.
For tensors $A, B$ of order $t$ with the same dimensions, the inner product
$\langle A, B \rangle$ is the entrywise inner product of the corresponding arrays.


\paragraph{Implicit Moment Tensor Estimation}
Roughly speaking, our algorithmic result shows that it is possible to efficiently
compute a {\em succinct approximate description} of tensors of the form
$M_m = \sum_{i=1}^{k} w_i v_i^{\otimes m}$, for non-negative $w_i$'s,
even for large (aka super-constant) values of the order parameter $m$, assuming
there exists a nice polynomial-size arithmetic circuit whose expected
output---under an appropriate sampleable distribution---is equal to $M_m$
and whose covariance (under the same distribution) is not too large.

Having a succinct description of a higher order tensor
begs the question of what exactly we can \emph{do} with this description;
as we cannot hope, for example, to efficiently decompress the description back into the full tensor.
Instead, we could hope that for the following:
for some reasonable class of functions $f:(\R^d)^{\otimes m} \to \R$,
our succinct description allows us to efficiently compute approximations to $f(M_m)$.
A natural class of such functions involves taking the inner product of $M_m$
with some tensor $T$ (of the same order and dimension). Of course, to have a hope of achieving this,
$T$ itself will need to have a concise representation---which for our purposes means that $T$
is the expected output of a nice arithmetic circuit on some sampleable random input distribution.
Fortunately, this limited means of querying $M_m$ turns out to be sufficient for many learning theoretic applications.
For further details on these applications, see Section~\ref{sssec:tech-learn}.

We are now ready to state an informal description of our algorithmic result in this context
(see Definition~\ref{def:seq-tensor} and Proposition~\ref{prop:main} for the detailed formal versions).

\begin{proposition}[Implicit Moment Tensor Computation, Informal] \label{prop:main-informal}
Let $d, k , m \in \Z_+$, $w_i \in \R_{\geq 0}$ and $v_i \in \R^d$ for all $i \in [k]$.
Consider the sequence of dimension-$d$ tensors $\{M_t\}$ with $M_t$ order-$t$, for $t\in\Z_+$,
defined by $M_t = \sum_{i=1}^{k} w_i v_i^{\otimes t}$.
Suppose that both assumptions below hold:
\begin{enumerate}[leftmargin=*, nosep]
\item For all positive integers $t \leq 2m$, with $t$ even or equal to $m$,
we have a reasonable way to sample from a distribution $\mathcal{S}_t$ over order-$t$ tensors
that has expectation $M_t$ and variance at most $V$.
\item We have a reasonable way to sample from a distribution $\mathcal{F}_m$ over order-$m$ tensors
that has expectation $T$ and variance at most $V$.
\end{enumerate}
Then there is an algorithm that draws $N$ such samples,
runs in time $\poly(N, d, k, m)$, and computes an approximation to
$\langle T, M_m \rangle$ that with high probability has absolute error
$$\frac{\poly\left(k, m, d, V\right) \left(1+\max_{i \in [k]}\|v_i\|_2\right)^{m}}{N^{1/4}} \;.$$
\end{proposition}

\paragraph{Discussion}
Some comments are in order. Recall that the goal is to efficiently compute a succinct
approximation to $M_m$ for a given $m \in \Z_+$. Formally speaking, the distribution over order-$t$ tensors
(with expected value $M_t$) that we are assumed to have sample access to is the output
of a nice arithmetic circuit (in particular, what we term a Sequential Tensor Computation; see Definition~\ref{def:seq-tensor})
on some sampleable distribution.
The algorithm of Proposition~\ref{prop:main-informal} outputs an approximation of the inner product
$\langle T, M_m \rangle$, for a tensor $T$ that is the expected output of a nice arithmetic circuit
on some sampleable distribution.



\paragraph{Comparison to Prior Work}
Proposition~\ref{prop:main-informal} gives an efficient method
to approximate higher moment tensors
of the form $\sum_{i=1}^k w_i v_i^{\otimes t}$ for non-negative $w_i$'s,
where the order $t$ is large,
so long as there are ``reasonable'' unbiased estimators
for these moments that can be efficiently sampled.
This result broadly generalizes
the work by Li and Liu~\cite{LL21-opt},
who introduced the idea of implicit moment tensor estimation in the context of clustering
mixtures of separated spherical Gaussians and, more broadly,
separated mixtures of translates of any known Poincare distribution $D$.
Roughly speaking, using the terminology introduced above
they gave an efficient algorithm for implicit moment tensor estimation only
for the {\em specific} tensors 
arising in their clustering
application---namely, 
the parameter moment tensors computed in a specific way using their so-called adjusted polynomials. 

The key conceptual contribution of our work
is the realization that the computational task of implicit moment tensor estimation
can be solved efficiently in a vastly more general setting, and thus
applied to several other learning problems.
This answers the main open question posed in~\cite{LL21-opt};
see, e.g.,~\cite{Liu24-TTIC}.

It turns out that the kind of sequence
of tensors $M_t = \sum_i w_i v_i^{\otimes t}$, that our technique requires,
arises in a variety of probabilistic models as Fourier or moment tensors.
We also generalize the way in which these tensors can be computed
from the expectation of an adjusted-polynomial to the output 
of any Sequential Tensor Computation (Definition~\ref{def:seq-tensor})---which is new
to this work.
Learning these moment tensors is thus a powerful tool and, as we have shown (Section~\ref{ssec:intro-apps}),
leads to more efficient learning algorithms for these models.
Beyond its generality, this broader view leads to quantitative improvements
even for the specific clustering application of~\cite{LL21-opt},
allowing us to cluster mixtures of spherical Gaussians under the optimal separation of $\sqrt{\log(k)}$;
see Theorem~\ref{thm:inf-cluster-GMMs} and Remark~\ref{rem:comp}.

\subsection{Technical Overview} \label{ssec:techniques}

\subsubsection{Proof Overview of Proposition~\ref{prop:main-informal}} \label{sssec:tech-impl}

The naive method to approximate  $\langle T, M_m \rangle$ is as follows:
Since $M_m$ is the expectation of $\mathcal{S}_m$, we can approximate
$M_m$ by drawing sufficiently many samples from this distribution and averaging them.
We can approximate $T$ similarly by sampling from $\mathcal{F}_m$, and 
computing the inner product of these approximations yields an approximation to the desired answer.

This naive approach suffers from two major issues that we need to overcome.
The first is statistical and the second is computational.

\paragraph{Statistical Challenges}
We start with the statistical aspect.
Note that both $T$ and $M_m$ are order-$m$ tensors in $d$ dimensions, i.e.,
they lie in a $d^m$-dimensional space.
Thus, in the absence of further structural properties, in
order to approximate them to error $\eps$ by sampling from $\mathcal{S}_m$ and $\mathcal{F}_m$, one
would need roughly $\Omega(d^m (V/\eps^2))$ samples,
where $V$ is the covariance upper bound.
Unfortunately, this sample upper bound is prohibitively large
for our learning applications---indeed, in the underlying applications, we need to take $m = \omega(1)$.\footnote{Concretely,
for learning mixtures of $k$ spherical Gaussians
one would need $m = \Omega(\log(k))$. For learning sums of
ReLU activations, we would need to select $m = \poly(1/\eps)$, where $\eps$ is the desired $L_2$-error.
Finally, for
learning $k$-mixtures of linear regressions with Gaussian noise of standard deviation $\sigma$,
one would need $m= \Omega(\log(1/\eps)/\sigma^2)$.}

In order to overcome this obstacle,
we need a method of reducing the dimension of the underlying space.
We can achieve this by considering
moment tensors $M_t = \sum_{i=1}^{k} w_i v_i^{\otimes t}$ of different orders.
Suppose, for example, that we knew {\em exactly} the moment tensor
$M_{2m} = \sum_{i=1}^k w_i v_i^{\otimes 2m} =
\sum_{i=1}^k w_i v_i^{\otimes m} \otimes v_i^{\otimes m}$.
When viewed as a $d^m \times d^m$ matrix, $M_{2m}$
will be rank-$k$ with image (span of its column vectors)
$W$ spanned by the $v_i^{\otimes m}$'s.
Thus, if we already knew $M_{2m}$, we could compute its image $W$
and then note that $M_m$ lies in $W$.
Then, instead of computing the full output tensors of
the $\mathcal{S}_m$ and $\mathcal{F}_m$---which
are objects lying in $(\R^d)^{\otimes m}$---we could compute their projections onto $W$.
Since $\dim(W) \leq k$ (as $M_{2m}$ is rank-$k$ by definition),
by standard concentration arguments it follows that
$O(k V/\eps^2)$ samples suffice to
compute these projections to $\ell_2$-error $\eps$.

\paragraph{Computational Challenges} 
The above discussion assumes that we have access to the tensor $M_{2m}$.
This begs the question of how we can compute $M_{2m}$, even approximately.
It turns out that this can be achieved via some form of bootstrapping.
In particular, suppose that we have computed
(an approximation to) $M_{2t}$, for some $t\leq m$.
We can write $M_{2t} = \sum_{i=1}^k w_i v_i^{\otimes t} \otimes v_i^{\otimes t}$.
When viewed as a $d^t \times d^t$ matrix, $M_{2t}$
has image spanned by the $v_i^{\otimes t}$'s.
Since the image is at most $k$-dimensional,
if $W_t$ is the span of the top-$k$ singular vectors of (our approximation to) $M_{2t}$,
then the $v_i^{\otimes t}$'s will all (approximately) lie in $W_t$.
This means that $M_{2t+2}$
will (approximately) lie in the space
$$W_t \otimes \R^d \otimes W_t \otimes \R^d \;.$$
As this space has dimension $\dim(W_t)^2 d^2 =O(k^2 d^2)$,
this allows a relatively fast  approximation to
$M_{2t+2}$. Thus, an approximation to $M_{2t}$ allows us to compute an approximation to
$M_{2t+2}$. A careful analysis of the error rates involved
resolves our statistical issues.

While the aforementioned ideas can be used to overcome sample complexity considerations,
significant {\em computational} challenges remain.
In particular, all of the vectors/tensors in question still lie in $d^t$ dimensional spaces;
so, even writing them down explicitly will take $\Omega(d^t)$ time.
We circumvent this obstacle by devising a {\em succinct representation}
of the relevant vector spaces {\em that suffices for the purpose of the underlying computations}.
In particular, by the above discussion, we can express $W_{t+1}$
(the span of the top-$k$ singular vectors of our approximation to $M_{2(t+1)}$)
as a subspace of $W_t \otimes \R^d$. By associating this ($k$-dimensional) subspace
with $\R^k$ (by picking an appropriate orthonormal basis/isometry),
we can then associate $W_{t+2}$
with a subspace of $\R^k \otimes \R^d$; and so on.
In particular, for each $t \leq m$,
we will have a map $\Phi_t: (\R^{d})^{\otimes t} \to \R^k$
given by the projection onto $W_t$ followed by an isometry.
These maps can then be efficiently constructed recursively from each other.

\subsubsection{From Proposition~\ref{prop:main-informal} to Learning Latent-variable Models}
\label{sssec:tech-learn}

In this subsection, we sketch how Proposition~\ref{prop:main-informal} is
leveraged to obtain our learning applications.
We start with the problem of learning sums of non-linear activations,
followed by our density estimation algorithms for MLRs and GMMs, and
concluding with our application on clustering GMMs.

\paragraph{PAC Learning Sums of Non-linear Activations}
Recall that the goal of this task is to approximate
a function of the form $F(x) = \sum _{i=1}^k w_i \sigma(v_i \cdot x)$,
where $w_i \geq 0$ and $\sigma$ is an appropriate (non-linear) activation function,
given access to random labeled examples $(x, F(x))$ with $x \sim N(0, I)$.
We proceed by attempting to learn the low degree part of the Fourier transform of $F$.
In particular, we have that
$$
F(x) = \sum_{m=0}^\infty \langle T_m, H_m(x)\rangle
$$
where $H_m$ is the normalized Hermite tensor (Definition~\ref{def:Hermite-tensor})
and $T_m = \E[F(G)H_m(G)]$. Our approach relies on approximating $F$ by the sum of the first $n$
many terms in the sum above. To do this, we note that the tensors $T_m$ are moment tensors
of approximately the kind desired. 
Namely, it can be shown (see Lemma~\ref{lem:Hermite-sum-ReLUs}) that
$$
T_m = c_{\sigma,m} \sum_{i=1}^k w_i v_i^{\otimes m} \;,
$$
for some known constants $c_{\sigma,m}$ 
(which are critically bounded away from $0$ when $m$ is even).
Given this, we can use Proposition \ref{prop:main-informal} directly to approximate
$\langle T_m, H_m(x)\rangle$ in time roughly $\poly(kd/\eps)2^{O(m)}$.
Doing this for all $m$ up to degree $n$
yields the desired approximation (see proof of Theorem~\ref{thm:formal-sum-ReLUs} for the detailed argument).

\paragraph{Density Estimation for MLRs and spherical GMMs}
Our algorithms for density estimation, both
for mixtures of spherical Gaussians and mixtures of linear regressions,
proceed similarly to the function approximation one above.
In particular, in order to do density estimation for a distribution $X$,
our goal will be to approximate the function $F(x) := X(x)/G(x)$, where $G$ is the standard Gaussian.
Once we have this, we can simulate $X$ via rejection sampling, by producing a sample $x\sim G$
and accepting with probability proportional to $F(x)$.

As above, $F(X)$ has a Taylor expansion given by
$$
F(x) = \sum_{n=0}^\infty \langle T_n, H_n(x)\rangle \;,
$$
where $T_n = \E[H_n(X)]$. This in turn needs to be related to an appropriate moment tensor.
For Gaussian mixtures, $T_n$ is already $\sum_{i=1}^k w_i \mu_i^{\otimes n}$.
For mixtures of linear regressions on the other hand, it needs to be computed indirectly
as $\sum_{i=1}^k w_i \beta_i^{\otimes n} = \E[(y^n/\sqrt{n!})H_n(X)]$,
and the Fourier coefficient is given by (see Lemma~\ref{MLRMomentCompLem})
$$
T_n = \frac{(n-1)!!}{\sqrt{n!}} \sum_{i=1}^k w_i \mathrm{Sym}\left(\left[ \begin{matrix} 0 & \beta_i\\ \beta_i^T & \|\beta_i\|_2^2 +\sigma^2-1\end{matrix} \right]^{\otimes n/2} \right) \;.
$$
The inner product of this quantity with $H_n(x)$ can be obtained indirectly
by taking appropriate inner products with the moment tensor.

\paragraph{Clustering Spherical GMMs}
The basic idea for our clustering application mirrors the approach of \cite{LL21-opt}:
given two samples $x$ and $x'$ from our mixture $X = \sum_{i=1}^k w_i N(\mu_i,I)$,
we want to reliably determine whether or not $x$ and $x'$ are from the same component or different ones.
If we can cluster samples by component, we can then easily learn each component.
We do this by considering the distribution $X' = (Y-Y')/\sqrt{2}$,
where $Y$ and $Y'$ are independent copies of $X$.
Here $X'$ is another mixture of Gaussians with means $\mu_i-\mu_j$.
We can determine whether $x$ and $x'$ are in the same component
by considering the size of the inner product of $(x-x')^{\otimes t}$
(actually in our case $H_t((x-x')/\sqrt{2})$) with the moment tensor of $X'$,
namely $\sum_{i,j} w_i w_j (\mu_i-\mu_j)^{\otimes t} = \E[H_t(X')]$; indeed,
this inner product will be large with high probability
if and only if $x$ and $x'$ come from different components of the mixture.

While the above is essentially the strategy employed by \cite{LL21-opt},
we devise a more efficient sequential tensor computation for estimating
the values of Hermite tensors (see Definition \ref{def:Hxy} and Lemma \ref{lem:cov-hxy}).
In particular, this allows us to efficiently consider order $t=\Theta(\log(k))$ tensors,
while \cite{LL21-opt} could only consider order up to $O(\log(k)/\log\log(k))$.
This quantitative improvement is responsible for our improvement in the separation parameter.

\subsection{Prior and Related Work} \label{ssec:related}
Here we record the most relevant prior work on the learning problems we study.

\vspace{-0.2cm}

\paragraph{Learning Mixtures of Linear Regressions}
Since their introduction~\cite{DV89, JJ94}, mixtures of linear regressions (MLRs)
have been broadly studied as a generative model for supervised data, and
have found applications to various problems, including trajectory clustering and phase retrieval;
see~\cite{Chen-thesis} for additional discussion and references.
A line of algorithmic work focused on the parameter estimation problem.
Specifically, a sequence of papers (see, e.g.,~\cite{ZJD16, LL18, KC19} and references therein)
analyzed non-convex methods,
including expectation maximization and alternating minimization.
These works established {\em local} convergence guarantees:
Given a sufficiently accurate solution (warm start),
these non-convex methods can efficiently boost this
to a solution with arbitrarily high accuracy.
The focus of our algorithmic results is to provide such a warm start,
which captures the complexity of the problem.
We note that the local convergence result of~\cite{LL18} applies for the noiseless case,
while the result of~\cite{KC19} can handle non-trivial
regression noise when the weights of the unknown mixture are known.

The prior works most closely related to ours are~\cite{LL18, CLS20, DK20-ag}.
The work of \cite{LL18} studied the noiseless setting
(corresponding to $\sigma = 0$) and provided a parameter estimation algorithm with sample complexity
and running time scaling exponentially with $k$.
Subsequently, the work of~\cite{CLS20} gave a sub-exponential time parameter estimation algorithm
for both the noiseless case and the case where $\sigma$ is small, namely $\sigma = O(\eps)$.
Specifically, if the weights are uniform and $\sigma = O(\eps)$,
their algorithm has sample and computational complexity
$\poly(d k /(\eps \Delta)) \, (k/\eps)^{\tilde{O}(k^{1/2}/\Delta^2)}$,
where $\Delta$ is the minimum pairwise separation between the regressors.
The fastest previously known algorithm~\cite{DK20-ag} for both
parameter and density estimation has sample and
computational complexity scaling quasi-polynomially in $k$.
Specifically, for density estimation (the task we focus on in this paper) 
the algorithm of ~\cite{DK20-ag}
has complexity $\poly(d, (k/\eps)^{\log^2(k)})$.

\vspace{-0.3cm}

\paragraph{Learning Mixtures of Spherical Gaussians}
Here we survey the most relevant prior work on learning
this distribution family both for density estimation and parameter estimation.

Density estimation for mixtures of high-dimensional spherical Gaussians
has been studied in a series of works~\cite{FOS:06, MoitraValiant:10,SOAJ14,
HardtP15, DKKLMS16, LiS17, AshtianiBHLMP18}.
The sample complexity of the problem for $k$-mixtures on $\R^d$, for variation distance error $\eps$,
is $\tilde{\Theta}(d k/\eps^2)$~\cite{AshtianiBHLMP18}.
Unfortunately, until fairly recently, all known algorithms had running times
exponential in number $k$; see~\cite{SOAJ14} and references therein.
The fastest density estimation algorithm prior to our work is from~\cite{DK20-ag}
and has complexity $\poly(d) (k/\eps)^{O(\log^2(k))}$.

In the related task of parameter estimation, the goal is to approximate the parameters of the components
(i.e., mixture weights and component means).
This task requires further assumptions, so that it is solvable with polynomial sample complexity
(even information-theoretically). The typical assumption in the literature is that pairwise separation
between the component means is bounded below by some parameter $\Delta>0$.
A long line of work since the late 90s steadily improved the separation
requirement~\cite{Dasgupta:99, AroraKannan:01, VempalaWang:02,
AchlioptasMcSherry:05, KSV08, BV:08, RV17-mixtures, HL18-sos, KSS18-sos, DKS18-list, LL21-opt}.

A related line of work studies the parameter estimation problem in a smoothed setting,
e.g., under the assumption that the means are linearly independent and the corresponding matrix
has bounded condition number.
Specifically,~\cite{HsuK13}
gave a polynomial-time learner for such instances. Interestingly, under such
assumptions, access to constant-order moments suffices. We emphasize that in the general case
that we focus on (i.e., without such assumptions) access to super-constant degree moments is necessary. 

The state-of-the-art algorithm for parameter estimation (in the general case)
prior to this paper is the work of~\cite{LL21-opt}: they gave a $\poly(d, k, 1/\eps)$ time
algorithm under near-optimal separation of $\Theta((\log (k/\eps))^{1/2+c})$.
We note that the assumption on spherical component covariances
is crucial for the latter result (and our Theorem~\ref{thm:inf-cluster-GMMs}).
Namely, a departure from sphericity leads to information-computation tradeoffs:
\cite{DKPZ23} gave evidence that quasi-polynomial runtime is inherent
for any $\polylog(k)$ separation, even if the Gaussian components have the same bounded covariance.

\vspace{-0.2cm}

\paragraph{Learning One-hidden-layer Networks}
Over the past decade, we have witnessed an explosion of research activity 
on provable algorithms for learning neural networks
in various settings, see, e.g.,~\cite{Janz15, SedghiJA16, DanielyFS16, ZhangLJ16,
ZhongS0BD17, GeLM18, GeKLW19, BakshiJW19, GoelKKT17,  GoelK19, VempalaW19, DKKZ20, GoelGJKK20, DK20-ag, ChenKM21,
ChenGKM22, CDGKM23, DK23-schur, ChenN24} for
some works on the topic.
Many of these works focused on parameter learning---the problem of
recovering the weight matrix of the data generating neural network.
PAC learning of simple classes of networks has been studied
as well ~\cite{GoelKKT17, GoelK19, VempalaW19, DKKZ20, ChenGKM22, CDGKM23, DK23-schur}.

The work of~\cite{GeLM18} studies the parameter learning of positive linear combinations of ReLUs
and other non-linear activations 
under the Gaussian distribution.
It is shown in~\cite{GeLM18}  that the parameters can be efficiently approximated,
if the weight matrix is full-rank with bounded condition number.
The complexity of their algorithm scales polynomially with the condition number.
\cite{BakshiJW19, GeKLW19} obtained efficient parameter learners for vector-valued
depth-$2$ ReLU networks under the Gaussian distribution. Similarly, the algorithms in these works
have sample complexity and running time scaling polynomially with the condition number.

In contrast to parameter estimation, PAC learning 
does not require any assumptions on the structure of the weight matrix. The PAC learning
problem for one-hidden-layer networks is information-theoretically solvable with polynomially
many samples. The question is whether a computationally efficient algorithm exists.
For the task of positive linear combinations of ReLUs studied in this work,
\cite{DKKZ20} gave a learner with 
computational complexity $\poly(dk/\eps)+(k/\eps)^{O(k^2)}$.
This runtime bound was improved to $\poly(d)(k/\eps)^{O(\log^2(k))}$ in~\cite{DK20-ag}, 
which was the state-of-the-art prior to our work.


\subsection{Organization}

The structure of the paper is as follows:
In Section~\ref{sec:prelims}, we provide the necessary definitions and technical facts.
In Section~\ref{sec:alg-general}, we prove our main algorithmic result on implicit tensor estimation.
Section~\ref{sec:apps} presents our learning algorithms for  sums of ReLUs,
mixtures of spherical Gaussians, and mixtures of linear regressions.

\section{Preliminaries} \label{sec:prelims}

\paragraph{Notation}
For $n \in \Z_+$, we denote by $[n]$ the set $\{1, 2, \ldots, n\}$.
For a vector $v \in \R^d$, let $\| v \|_2$ denote its Euclidean norm.
We denote by $x \cdot y$ the standard inner product between $x, y \in \R^d$.
We will denote by $\delta_0$ the Dirac delta function
and by $\delta_{i, j}$ the Kronecker delta.
Throughout the paper, we let $\otimes$ denote the tensor/Kronecker product.
For a vector $x \in \R^d$, we denote by $x^{\otimes m}$ the $m$-th order tensor power of $x$.

Fix some $t \in \Z_+$. For vector spaces $V_i$, $i \in [t]$,
an order-$t$ tensor on $\bigotimes_{i=1}^t V_i$ is simply an element $A \in \bigotimes_{i=1}^t V_i$.
If $A, B \in \bigotimes_{i=1}^t V_i$, for some inner product spaces $V_i$, $i \in [t]$,
then we use  $\langle A, B \rangle$ to denote the inner product of $A$
and $B$ {\em induced by the inner products on $V_i$}.
Namely, using the unique inner product structure on $\bigotimes_{i=1}^t V_i$
so that
$\langle v_1\otimes \cdots \otimes v_t, w_1\otimes \cdots \otimes w_t\rangle
= \prod_{i=1}^t \langle v_i, w_i\rangle$ for all vectors $v_i,w_i \in V_i$.
For the special case that $V_i = \R^{d_i}$, $d_i \in \Z_+$,
a tensor $A$ is defined by a $t$-dimensional array with real entries $A_{\alpha}$,
where $\alpha = (\alpha_1, \ldots, \alpha_t)$ with $\alpha_i \in [d_i]$.
For tensors $A, B$ of order $t$ with the same dimensions, the inner product
$\langle A, B \rangle$ is the entrywise inner product of the corresponding arrays.
We also use $\|A\|_2 = \langle A, A \rangle^{1/2}$ for the corresponding $\ell_2$-norm.


The covariance of an order-$t$, dimension $d$ tensor-valued random variable $T \in (\R^d)^{\otimes t}$
is defined as the tensor in $(\R^d)^{\otimes 2t}$ defined as
$\Cov[T]: = \E[T \otimes T] - \E[T] \otimes \E[T]$, where the expectation operation is applied entry-wise.
Similarly, the corresponding ``second moment''
is the tensor $ \E[T \otimes T]$.
For an order-$t$, dimension-$d$ tensor $W$, consider the real-valued random variable
$\langle W, T \rangle$ and its corresponding variance, $\var[\langle W, T \rangle]$.
We say that $\Cov[T]$ is bounded above by $C$, for some $C>0$,
if $\max_{\|W\|_2 = 1} \var[\langle W, T \rangle] \leq C$, i.e., if the variance of $T$ in any (normalized)
direction is at most $C$. The analogous definition applies for the second moment $ \E[T \otimes T]$.

For a space $W$, we use $I_W$ for the identity matrix on $W$. If $W = \R^d$, we use the notation $I_d$.

We will denote by $N(0, I_d)$ the $d$-dimensional Gaussian distribution
with zero mean and identity covariance; we will use $N(0, I)$ when the underlying dimension
will be clear from the context. We will use $N(0, 1)$ for the univariate case.
For a random variable $X$ and $p \geq 1$,
we will use $\|X\|_p \eqdef \E[|X|^p]^{1/p}$ to denote its {\em $L_p$-norm}.

\paragraph{Hermite Analysis and Concentration} \label{ssec:hermite}
Consider $L_2(\R^d, N(0, I))$, the vector space of all
functions $f : \R^d \to \R$ such that $\E_{x \sim N(0, I)}[f(x)^2] <\infty$.
This is an inner product space under the inner product
$\langle f, g \rangle = \E_{x \sim N(0, I)} [f(x)g(x)] $.
This inner product space has a complete orthogonal basis given by
the \emph{Hermite polynomials}.
In the univariate case, we will work with normalized Hermite polynomials
defined below.

\begin{definition}[Normalized Probabilist's Hermite Polynomial]\label{def:Hermite-poly}
For $k\in\N$, the $k$-th \emph{probabilist's} Hermite polynomial
$He_k:\R\to \R$ is defined as
$He_k(t)=(-1)^k e^{t^2/2}\cdot\frac{\mathrm{d}^k}{\mathrm{d}t^k}e^{-t^2/2}$.
We define the $k$-th \emph{normalized} probabilist's Hermite polynomial
$h_k:\R\to \R$ as
$h_k(t)=He_k(t)/\sqrt{k!}$.
\end{definition}

\noindent Note that for $G\sim N(0,1)$ we have $\E[h_n(G)h_m(G)] = \delta_{n,m}$.
We will use multivariate Hermite polynomials in the form of
Hermite tensors. We define the normalized Hermite tensor as follows,
in terms of Einstein summation notation.

\begin{definition}[Normalized Hermite Tensor]\label{def:Hermite-tensor}
For $k\in \N$ and $x \in V$ for some inner produce space $V$,
we define the $k$-th Hermite tensor as
\[
(H_k^{(V)}(x))_{i_1,i_2,\ldots,i_k}:=\frac{1}{\sqrt{k!}}\sum_{\substack{\text{Partitions $P$ of $[k]$}\\ \text{into sets of size $1$ and $2$}}}\bigotimes_{\{a,b\}\in P}(-I_{i_a,i_b})\bigotimes_{\{c\}\in P}x_{i_c} \;,
\]
where $I$ above denotes the identity matrix over $V$. Furthermore, if $V=\R^d$,
we will often omit the superscript and simply write $H_k(x)$.
\end{definition}

\noindent
We will require a few properties that follow from this definition.
First, note that if $V$ is a subspace of $W$,
then $H^{(V)}_k(\mathrm{Proj}_V(x)) = \mathrm{Proj}_V^{\otimes k}H_k^{(W)}(x)$.
Applying this when $V$ is the one-dimensional subspace spanned by a unit vector $v$
gives that
$\langle H_k(x), v^{\otimes k} \rangle = h_k(v \cdot x)$.
We will also need to know that the entries of $H_k(x)$
form a useful Fourier basis of $L^2(\R^d,N(0,I))$.
In particular, for non-negative integers $m$ and $k$, we have that
$\E_{x \sim N(0, I)}[H_k(x) \otimes H_{m}(x)]$ is $0$
if $m\neq k$ and  $\mathrm{Sym}_k(I_{d^k})$, if $m=k$,
where $\mathrm{Sym}_k$ is the symmetrization operation over the first $k$ coordinates.
From this we conclude that if $T$ is a symmetric $k$-tensor,
then $\E_{x \sim N(0, I)}[\langle H_k(x), T \rangle H_m(x)]$
is $0$ if $m\neq k$ and $T$ if $m=k$.

We will also require the following fact, where
$\mathrm{Sym}$ denotes the symmetrization operator
that averages a tensor over all permutations of its entries.

\begin{fact} \label{fact:Ht-mean-cov}
We have that $\E_{X \sim N(\mu,I)}[H_n(X)] = \mu^{\otimes n}/\sqrt{n!}$ and
$\Cov_{X \sim N(0,I)}[H_n(X)] = \mathrm{Sym}(I_{d^n})$.
\end{fact}

For a polynomial $p: \R^d \to \R$, we will use $\|p\|_r \eqdef \E_{x \sim N(0, I)} [|p(x)|^r]^{1/r}$,
for $r \geq 1$. We recall the following well-known hypercontractive
inequality~\cite{Bon70,Gross:75}:
\begin{fact} \label{thm:hc}
Let $p: \R^d \to \R$ be a degree-$k$ polynomial and $q>2$.  Then
$\|p\|_q \leq (q-1)^{k/2} \|p\|_2$.
\end{fact}

\section{Main Result: Proof of Proposition~\ref{prop:main-informal}} \label{sec:alg-general}

The structure of this section is as follows:
Section~\ref{ssec:formal} presents our formal framework
and the detailed statement of Proposition~\ref{prop:main-informal}.
In Section~\ref{ssec:rp}, we establish some technical results that are necessary
for our implicit tensor estimation algorithm and its analysis.
Finally, in Section~\ref{ssec:alg-analysis} we give the pseudocode of our algorithm
and prove its correctness.

\subsection{Formal Framework and Statement of Main Result} \label{ssec:formal}

To state our result, we first need to formalize what we mean
by a ``nice'' arithmetic circuit to compute our tensors. The key point is that
this method of computation must be compatible with our method
for implicitly representing the subspaces $W_t$ (described in Section~\ref{sssec:tech-impl}):
namely, repeatedly tensoring with $\R^d$ and then applying
some linear transformation back to $\R^k$. Compatibly with this,
we can add two elements of a single $W_t$ or multiply one by a scalar.
We can also take the tensor product of a tensor in $W_t$ with an element of $\R^d$
and map the result to $W_{t+1}$. Unfortunately, we cannot easily perform other operations
like take the tensor product of two elements of $W_t$ and find its image in $W_{2t}$.
The notion of a \emph{Sequential Tensor Computation} below amounts
to an arithmetic circuit that performs only operations compatible with these computations.

\begin{definition}[Sequential Tensor Computation] \label{def:seq-tensor}
A {\em sequential tensor computation} (STC) is an arithmetic circuit $\mathcal{S}_t$
that takes as input $\mathcal{I}$ a set of scalars and vectors in $\R^d$ and outputs
an order-$t$ tensor of dimension $d$,
denoted by $\mathcal{S}_t(\mathcal{I}) \in (\R^d)^{\otimes t}$,
for some $t \in \Z_+$,
by applying the following operations:
\begin{itemize}
\item[(i)] Multiplication of a tensor, vector, or scalar by a scalar.
\item[(ii)] Addition of two tensors, vectors or scalars of the same order and dimension.
\item[(iii)] Tensor product of an order-$k$ tensor, for some $k < t$,
by a vector (in the last component)
to obtain an order-$(k+1)$ tensor.
\end{itemize}
We let the order of $\mathcal{S}_t$ be used to denote the
order of the output tensor, and the size of $\mathcal{S}_t$
to denote the size of the underlying arithmetic circuit.
\end{definition}

\noindent We note that an STC is a general way to produce higher-order tensors
that is compatible with the kinds of computations that we need to perform on them
(see Lemmas~\ref{lem:rpp-stc} and~\ref{lem:tensor-rpp}) in the context of learning applications.

Our main algorithmic result (i.e., the formal version of Proposition~\ref{prop:main-informal}) is the following:

\begin{proposition}[Main Algorithmic Result on Implicit Tensor Computation] \label{prop:main}
Let $d, k \in \Z_+$, $w_i \in \R_{\geq 0}$ and $v_i \in \R^d$ for all $i \in [k]$.
Consider the sequence of  dimension-$d$ tensors $\{M_t\}$ with $M_t$ order-$t$, for $t\in\Z_+$,
defined by $M_t = \sum_{i=1}^{k} w_i v_i^{\otimes t}$.

Let $m \in \Z_+$.
Suppose that for all positive integers $t \leq 2m$, with $t$ even or equal to $m$,
there exists an STC $\mathcal{S}_t$ of order-$t$ and
size at most $S$ with the following property:
when given as input $\mathcal{I}$
a sample drawn from some efficiently samplable distribution $\mathcal{D}_t$,
the output tensor $\mathcal{S}_t(\mathcal{I})$, $\mathcal{I} \sim \mathcal{D}_t$,
has mean $M_t$ and covariance bounded above by $V$, for some $V>0$.
Let $\mathcal{F}_m$ be another STC
of order-$m$ and size at most $S$, whose input is partitioned
as $(\mathcal{X}, \mathcal{Y})$, with the following properties:
when $\mathcal{X} \sim X$ and $\mathcal{Y} \sim \mathcal{D}'$, where $\mathcal{D'}$ is efficiently samplable,
the tensor $\mathcal{F}_m(\mathcal{X}, \mathcal{Y})$ has expectation over
$\mathcal{Y} \sim \mathcal{D'}$
\begin{equation}\label{eqn:T}
T(\mathcal{X}): = \E_{\mathcal{Y} \sim \mathcal{D'}}[\mathcal{F}_m(\mathcal{X}, \mathcal{Y})]
\end{equation}
and second moment over both $\mathcal{X}$ and $\mathcal{Y}$ bounded by $V$.

There is an algorithm that given $S, V, k$,
sample access to $\mathcal{D}_t, \mathcal{D'}$, and $\tau>0$,
it has the following guarantee:
on input a single sample $\mathcal{X} \sim X$,
the algorithm draws $N$
samples from each of $\mathcal{D}_t, \mathcal{D'}$,
runs in $\poly(N, S, d)$ time,
and with probability $1-\tau$ (over the samples drawn from $\mathcal{D}_t, \mathcal{D'}$)
outputs an approximation $\mathcal{A}$ to $\langle T(\nnew{\mathcal{X}}), M_m \rangle$ such that
\begin{equation} \label{eqn:error-bound}
\E_{\mathcal{X} \sim X}[(\mathcal{A}-\langle T(\nnew{\mathcal{X}}), M_m \rangle)^2] \leq
\frac{\poly\left(k, m, d, V, 1/\tau, 1+\littlesum_{i=1}^k w_i \right) \left(1+\max_{i \in [k]}\|v_i\|_2\right)^{2m}}{N^{1/2}} \;.
\end{equation}
\end{proposition}

\paragraph{Discussion} In words, we show that, with high probability over the samples
drawn from $\mathcal{D}_t$ and $\mathcal{D}'$,
the expected squared error of our estimator over a random choice of \nnew{$\mathcal{X} \sim X$}
is bounded as shown in \eqref{eqn:error-bound} above. In particular, our algorithm can be thought of
as taking samples from $\mathcal{D}_t$ and $\mathcal{D}'$ as input and
returning an evaluator for a function of \nnew{$\mathcal{X}$} that with high probability
is close in $L_2$ to the function $\langle T(\nnew{\mathcal{X}}),M_m \rangle$.

\begin{remark} \label{rem:gen}
{\em It turns out that this kind of generalization described above is necessary for our applications---rather than
simply finding the inner product of $M_m$ with a particular tensor $T$.
This is because, for example, we might want to compute $\langle H_m(x),M_m\rangle$
for the Hermite tensor $H_m$ evaluated at a random $x$.
While we could simply let $T=H_m(x)$ for some given $x$,
that would run into the problem that $H_m(x)$ is itself quite large---so a naive application of our result
would give large error. However, as the second moment of $H_m(x)$ is bounded,
our result says that the mean squared error over a random choice of $x$ is small.}
\end{remark}

\noindent A couple of additional remarks are in order. First, we did not attempt to optimize the polynomial dependence
on the relevant parameters, as this would not affect the qualitative aspects of our
learning theory applications. Second, the dependence
on the failure probability $\tau$ can be made logarithmic, i.e.,  $\log(1/\tau)$,
by running the algorithm $O(\log(1/\tau))$ times and taking a median;
such a dependence is not required in our applications.


\subsection{Recursive Pseudo-projections and their Properties} \label{ssec:rp}

In order to describe the subspaces $W_t$, we will need a general tool
that we call a \emph{recursive pseudo-projection}, which we develop below. 
We start by defining the notion of a pseudo-projection.

\begin{definition}[Pseudo-projection] \label{def:pp}
A linear transformation $A: X \to W$ between finite dimensional inner product spaces
$X, W$
is called a {\em pseudo-projection} if $AA^{\top} = I_W$.
\end{definition}

\noindent The following lemma establishes two basic properties of pseudo-projections that we will require.

\begin{lemma}[Properties of Pseudo-projections]\label{lem:pp-basics}
The following properties hold:
\begin{enumerate}[leftmargin=*]
\item The composition of pseudo-projections is a pseudo-projection.
\item A pseudo-projection $A: X \to W$ can be written as a projection
$P: X \to U$, for some subspace $U$ of $X$, composed with an isometry from $U$ to $W$.
\end{enumerate}
\end{lemma}
\begin{proof}
Property (1) follows immediately from the definition.
For property (2), note that
the right singular values of $A$
are equal to $1$, and its left singular values are $1$ and $0$.
Therefore, $A^{\top} A$ is a projection onto some subspace
$U$ of $X$ of the same dimension as $W$.
Since $A^{\top}$ maps $W$ to $U$ and preserves norms, it is an isometry.
Since $A$, mapping $U$ to $W$, is its inverse, it follows that
$A: U \to W$ (the restriction of $A$ on $U$) is also an isometry.
Note that $A: X \to W$ can be written as
$A = A(A^{\top}A) = (AA^{\top})A$,
which is the composition of a projection onto $U$ and an isometry.
\end{proof}

\noindent In order to describe the {\em recursive} projections, we introduce the following definition:

\begin{definition}[Recursive Pseudo-projection] \label{def:rpp}
For $d, t,  n \in \Z_+$, a {\em $(d, t, n)$-recursive pseudo-projection}
is a pseudo-projection $\Phi: (\R^d)^{\otimes t} \to \R^n$,
defined as follows. There exist positive integers $n_0,n_1,\ldots,n_t$
with $n_0=1$ and $n_t = n$ and base pseudo-projections $\phi_i: \R^{n_{i-1}} \otimes \R^d \to \R^{n_i}$,
$i \in [t]$, such that the following holds: We define
$\Phi: \R^1 \otimes (\R^d)^{\otimes t} \to \R^n$
by starting with a tensor $T_0$ in $\R^{n_0} \otimes (\R^d)^{\otimes t}$,
applying $\phi_1$ to the first two components of $T_0$
to get a tensor $T_1$ in $\R^{n_1} \otimes (\R^d)^{\otimes (t-1)}$, then applying
$\phi_2$ to the first two components of $T_1$
to get a tensor $T_2$ in $\R^{n_2} \otimes (\R^d)^{\otimes (t-2)}$;
and more generally applying $\phi_{i}$, $i \in [t]$, to the first two components of $T_{i-1}$,
which is a tensor in $\R^{n_{i-1}} \otimes (\R^d)^{\otimes (t-i+1)}$,
to get a tensor $T_{i}$ in $\R^{n_{i}} \otimes (\R^d)^{\otimes (t-i)}$.
We say that the recursive pseudo-projection $\Phi$ has order $t$ and size $\max_i n_i$.
\end{definition}

It is easy to verify that a recursive pseudo-projection is itself a pseudo-projection.
We will require a few basic properties of recursive pseudo-projections.
First, we show that a recursive pseudo-projection can efficiently
be applied to the output of a sequential tensor computation:

\begin{lemma}[Efficient Computation]\label{lem:rpp-stc}
Let $\Phi: (\R^d)^{\otimes t} \to \R^n$ be an order-$t$ recursive pseudo-projection
whose defining pseudo-projections $\phi_i$ are given explicitly as matrices.
Let $\mathcal{S}$ be a sequential tensor computation returning an order-$t$ and dimension-$d$ tensor.
Then there is an algorithm that given input vectors $v_i$ for $\mathcal{S}$
returns $\Phi(\mathcal{S}(\{v_i\}))$ and runs in time $\poly(d, \size(\Phi), \size(\mathcal{S}))$.
\end{lemma}
\begin{proof}
At any point in the circuit defining $\mathcal{S}$
if there is an order-$s$ tensor, we can compute the first $s$ steps of
$\Phi$ applied to that value.
The goal is to for every wire in the circuit computing $\mathcal{S}$ to compute the value
of the appropriate iterate of the operation defining $\Phi$ applied to the value being
carried by that wire in the circuit.
We claim that this can be computed efficiently via dynamic programming
(see Figure~\ref{fig:DP} for an illustration).
We start with the wires at the beginning of the circuit
and work our way towards the end. It suffices to verify that if we know
the appropriate values for the input wires,
we can efficiently compute the corresponding value of the output wire.
This involves the operations of tensor summation,
multiplication of a tensor by a scalar sum,
and tensor product of tensor with a vector in the last coordinate.
It is straightforward to see that any of the allowed operations in $\mathcal{S}$ are compatible
with this computation and that it runs in polynomial time in the relevant parameters.
\end{proof}

\begin{figure}
\begin{center}
\epsfig{file=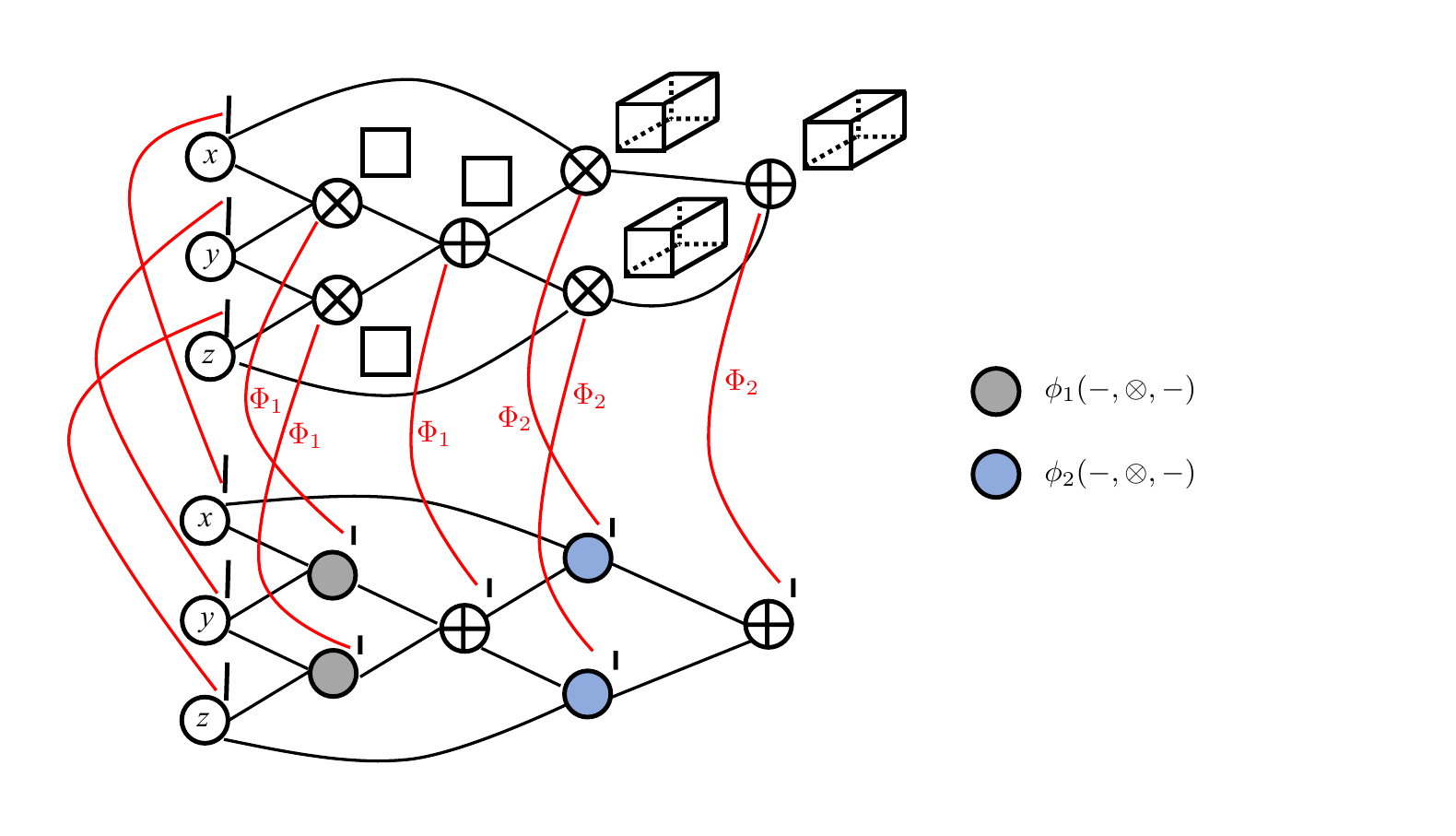, width=20cm}
\end{center}
\caption{The upper circuit represents the original Sequential Tensor Computation (STC).
The markings next to the circuit elements represent the complexity of the elements being dealt with.
The original inputs are vectors (lines). Then many of the middle components deal with $2$-tensors (squares),
and finally the circuit elements on the right deal with $3$-tensors (cubes). In an STC of higher degree,
one might need to use even higher degree tensors which quickly become computationally unmanageable.
Fortunately, we are able to simulate this computation using the circuit below,
keeping track of the projected values of each circuit element.
This keeps the computational complexity low by reducing everything
to vectors of dimension at most $\dim(\Phi)$ (represented by short lines).}
\label{fig:DP}
\end{figure}

\noindent We also need to be able to efficiently take tensor products of recursive pseudo-projections:

\begin{lemma}[Efficient Tensorization] \label{lem:tensor-rpp}
Given recursive pseudo-projections
$\Phi: (\R^d)^{\otimes t} \to \R^n$ and
$\Theta: (\R^d)^{\otimes s} \to \R^m$ along with their defining pseudo-projections
$\phi_i, \theta_j$, given explicitly as matrices, one can efficiently
construct the recursive pseudo-projection
$\Phi \otimes \Theta: (\R^d)^{\otimes (t+s)} \to \R^{n m}$ of size
$\max(\size(\Phi), n \; \size(\Theta))$.
\end{lemma}
\begin{proof}
The proof follows by using as the base pseudo-projections $\phi_1, \phi_2, \ldots ,\phi_t$,
$I_n \otimes \theta_1, \ldots ,I_n \otimes \theta_s$.
\end{proof}

\subsection{Algorithm and Analysis} \label{ssec:alg-analysis}

\noindent We are now ready to present the pseudo-code of the algorithm establishing Proposition~\ref{prop:main}.

\bigskip


\fbox{\parbox{6.1in}{
{\bf Algorithm} {\tt Implicit-Moment-Estimation}\\
\vspace{-0.2cm}

{\bf Input:}  Parameters $k, d, m \in \Z_+$; STCs $\mathcal{S}_t, \mathcal{F}_m$,
with size bound $S \in \Z_+$ and second moment bound $V \in \R_+$;
$N$ i.i.d.\ samples from $\mathcal{D}_t$, $\mathcal{D}'$,
single sample $\mathcal{X} \sim X$; failure probability $\tau$.\\
\vspace{-0.4cm}

{\bf Output:} A real number approximating $\langle M_m, T(\mathcal{X}) \rangle $,
where $M_m = \sum_{i=1}^k w_i v_i^{\otimes m}$ and
$T(\mathcal{X})= \E_{\mathcal{Y} \sim \mathcal{D'}}[\mathcal{F}_m(\mathcal{X}, \mathcal{Y})]$,
as defined in \eqref{eqn:T}.

\begin{enumerate}

\item Let $N$ be a positive integer quantifying the number of samples drawn
from the relevant distributions in each step.

\item Let $A_{1}$ be the average of $N$ runs of $\mathcal{S}_2$
on $N$ independent samples drawn from $\mathcal{D}_2$, thought of as a $d \times d$ matrix.

\item Let $W_{1}$ be the span of the top-$k$ singular vectors of $A_{1}$.
Let $\phi_1 : \R^1 \otimes \R^d \to \R^k$ be the projection of $\R^d$
onto $W_{1}$ composed with an isometry from $W_{1}$ to $\R^k$.

\item For $r = 1$ to $m-1$:
\begin{enumerate}
\item Let $\Phi_r: (\R^d)^{\otimes r} \to \R^k$ be the order-$r$
recursive pseudo-projection given by $\phi_1, \ldots ,\phi_{r}$.

\item Let $\Phi'_{r}: (\R^d)^{\otimes (2r+2)} \to \R^{k^2 d^2}$
be the recursive pseudo-projection $\Phi_r \otimes I_d \otimes \Phi_r \otimes I_d$.

\item Let $A_{r+1}$ be the average of $N$ copies of $\Phi_{r}'$ applied to the outputs
of $\mathcal{S}_{2r+2}$ on $N$ independent samples
drawn from $\mathcal{D}_{2r+2}$, using Lemma~\ref{lem:rpp-stc} to compute efficiently,
thought of as a $dk \times dk$ matrix.

\item Let $W_{r+1} \subset \R^{dk}$ be the span
of the top-$k$ singular vectors of $A_{r+1}$.

\item Let $\phi_{r+1}$ be the composition of the projection of $\R^k \otimes \R^d \to W_{r+1}$
with an isometry from $W_{r+1}$ to $\R^k$.
\end{enumerate}

\item  Let $\Phi_m: (\R^d)^{\otimes m} \to \R^k$
be the recursive pseudo-projection given by $\phi_1, \ldots ,\phi_{m}$.

\item Let $A$ be the average of $N$ copies of $\Phi_m$ applied to the output of $\mathcal{S}_m$
on $N$ independent samples drawn from $\mathcal{D}_m$,
using Lemma~\ref{lem:rpp-stc} to compute efficiently.

\item Let $B(\nnew{\mathcal{X}})$ be the average of $N$ copies of $\Phi_m$ applied to
the output of $\mathcal{F}_m$ on the single sample $\nnew{\mathcal{X}}$ and
$N$ independent samples from $\mathcal{D}'$,
using Lemma~\ref{lem:rpp-stc} to compute efficiently.

\end{enumerate}

Return $\langle A, B(\nnew{\mathcal{X}}) \rangle$.

}}

\newpage

\paragraph{Proof of Proposition~\ref{prop:main}}
By the definition of the $A_r$'s in the algorithm pseudocode we have that
(i) $A_1$ will be close to $M_2$, and (ii) $A_{r+1}$, $r \in [m-1]$,
will be close to $\Phi_{r}'(M_{2r+2})$ with high probability. Specifically,
with probability at least $1-\tau/4$ over the
samples drawn from the $\mathcal{D}_t$'s in the various steps,
we will have $\|A_1 - M_2\|_2 \leq \delta$
and $\|A_{r+1}- \Phi_{r}'(M_{2r+2})\|_2 \leq \delta$,
for all $r \in [m-1]$, where
\begin{equation}\label{eqn:delta}
\delta  \eqdef O\left(\frac{m \, d \, k\, V}{\sqrt{N \, \tau}} \right) \;.
\end{equation}
We will condition on the event that all these approximations hold in the subsequent analysis.

Let $\Phi_r: (\R^d)^{\otimes r} \to \R^k$ be
the recursive pseudo-projection given by $\phi_1,\ldots,\phi_r$.
By the second statement of Lemma~\ref{lem:pp-basics},
$\Phi_r$ is the composition of the projection of
$(\R^d)^{\otimes r}$ onto some subspace $U_r$ of $(\R^d)^{\otimes r}$
composed with an isometry.

\medskip

For $i \in [k]$ and $r \in [m]$, we will denote
$x_{i,r} \eqdef \sqrt{w_i} v_i^{\otimes r}$. With this notation, we can write
 $M_{2r} = \sum_{i=1}^{k} x_{i,r}^{\otimes 2}$.
Let
$$
\eta_r \eqdef \max_{i \in [k]}  \mathrm{dist}(x_{i,r}, U_{r}) =
\max_{i \in [k]} \left\|x_{i, r}  - \proj_{U_r}(x_{i, r})  \right\|_2 \;,$$
where we recall that
$U_r$ is the subspace of $(\R^{d})^{\otimes r}$ corresponding to $\Phi_r$.

Controlling the size of the $\eta_r$ will be vital to our analysis. 
We begin by establishing the following recursive bound: 

\begin{lemma} \label{lem:eta-r}
For all $r\in [m-1]$, it holds
$\eta_{r+1} \leq \eta_r \max_i \|v_i\|_2 + O(\sqrt{\delta})$.
\end{lemma}

\begin{proof}
By the definition of $\Phi_{r}'$ and $M_{2r+2}$, we can write:
$$\Phi_{r}'(M_{2r+2}) = \sum_{i=1}^{k} (\underbrace{\Phi_r(x_{i,r}) \otimes v_i}_{y_{i,r}})^{\otimes 2}
= \sum_{i=1}^{k} y_{i,r}^{\otimes 2} \;.$$
Moreover, we have that
$(\Phi_r \otimes I_d)^{\top} y_{i,r} = \Phi_r^\top(\Phi_r(x_{i,r}))\otimes v_i = \mathrm{Proj}_{U_r}(x_{i,r})\otimes v_i$.
Therefore,
\begin{equation} \label{eqn:x-y-dist}
\left\| (\Phi_r \otimes I_d)^{\top} y_{i,r} - x_{i,r+1} \right\|_2 \leq \eta_r   \| v_i\|_2 \;.
\end{equation}
There are three relevant objects for the analysis:
$M_{2r+2}$,
$$M^{\ast} \eqdef (\Phi_{r}')^{\top} \Phi_{r}' M_{2r+2}
= \sum_{i=1}^k (\Phi_{r}')^{\top}\Phi_r'(y_{i,r}^{\otimes 2}) \;,$$
and  $A^{\ast} \eqdef (\Phi_{r}')^{\top} A_{r+1}$.

We note that
$$\|A^{\ast} - M^{\ast}\|_2 = \| A_{r+1} - \Phi_{r}' M_{2r+2} \|_2 \leq \delta \;.$$
Recall that $\eta_{r+1} = \max_{i \in [k]} \mathrm{dist}(x_{i,r+1}, U_{r+1})$.
For all $i \in [k]$,
$$\mathrm{dist}(x_{i,r+1}, U_{r+1}) \leq \|x_{i,r+1} - (\Phi_r \otimes I_d)^{\top}y_{i,r}\|_2
+ \mathrm{dist}((\Phi_r \otimes I_d)^{\top}y_{i,r}, U_{r+1}) \;.$$
It follows from \eqref{eqn:x-y-dist} that
$$ \|x_{i,r+1} - (\Phi_r \otimes I_d)^{\top}y_{i,r}\|_2 \leq \eta_r \|v_i\|_2 \;.$$
Moreover, note that $\mathrm{dist}((\Phi_r \otimes I_d)^{\top}y_{i,r}, U_{r+1})$
is the distance from $(\Phi_r \otimes I_d)^{\top}y_{i,r}$ to the span
of the top-$k$ singular vectors of $A^{\ast}$.
Since $A^{\ast}$ is $O(\delta)$-close to
$M^{\ast}= \sum_{i=1}^k (\Phi_{r}')^{\top} y_{i,r}^{\otimes 2}$
(which is rank-$k$), we can show that
$\mathrm{dist}((\Phi_r \otimes I_d)^{\top}y_{i,r}, U_{r+1})  = O(\sqrt{\delta})$.
We prove this as follows.
Let $y:= (\Phi_r \otimes I_d)^{\top}y_{i,r} = u + v$
with $u \in U_{r+1}$ and $v$ orthogonal.
The distance in question is now just $\|v\|_2$.
Note that $v^{\top} M^{\ast} v \geq  \langle v, y \rangle ^2 \geq \|v\|_2^4$.
On the other hand $v^{\top} A^{\ast} v$ must be relatively small.
This is because $v$ is orthogonal to the top-$k$ singular vectors of $A^{\ast}$,
and so it is at most $\|v\|_2^2 \lambda_{k+1}$,
where $\lambda_{k+1}$ is the $(k+1)$-st singular value.
We have that
$$\lambda_{k+1}\leq  \sup_{w \perp (\Phi_r \otimes I_d)^{\top}y_{i,r}, 1\leq i \leq k, \|w\|_2=1} w^{\top} A^{\ast} w =
\sup_{w \perp (\Phi_r \otimes I_d)^{\top}y_{i,r}, 1\leq i \leq k, \|w\|_2=1} w^{\top} M^{\ast} w + O(\delta) = O(\delta) \;.$$
So $v^{\top} M^{\ast} v \geq \|v\|_2^4$, and $v^{\top} A^{\ast} v = O(\delta) \|v\|_2^2$, which gives that
the difference is
$$v^{\top} (A^{\ast}-M^{\ast}) v = O(\delta)\|v\|_2^2 \;.$$
Putting these together implies that $\|v\|_2 = O(\sqrt{\delta})$.

Thus, we have that $\eta_{r+1} = O(\eta_r \max_i \|v_i\|_2 + \sqrt{\delta})$,
proving Lemma~\ref{lem:eta-r}.
\end{proof}

\medskip

Lemma~\ref{lem:eta-r} allows us to establish an appropriate
upper bound on $\eta_m$, which suffices to show that
$\Phi_m^{\top} A$ is close to $M_m$ with high probability.
This in turns gives that $\langle A, B(X) \rangle$
close to $\langle M_m, T(X) \rangle$ in the mean squared sense.
We present the detailed argument below.

We start by using Lemma~\ref{lem:eta-r} to show an upper bound of $\eta_m$.
To do this, we need an upper bound on
$\eta_1 = \max_{i \in [k]} \mathrm{dist}(x_{i, 1}, U_1)$.
Recall that $x_{i,1} = \sqrt{w_i} v_i$ and note that  by construction it holds
$U_1 = W_1$,
where $W_1$ is the span of the top-$k$ singular vectors of $A_1$.

We will use the fact that $\|A_1 - M_2\|_2 \leq \delta$.
In the limit, when $\delta \rightarrow 0$, we have that
$A_1 = M_2 =  \sum_{i=1}^k x_{i,1} x_{i,1}^{\top}$,
in which case $W_1$ is the span of the $x_{i, 1}$'s and $\eta_1=0$.
For the case of $\delta>0$, a standard argument using Weyl's inequality
(similar to the one used in the proof of Lemma~\ref{lem:eta-r})
shows that $\|x_{i, 1} - \proj_{W_1}(x_{i, 1})\|_2  = O(\sqrt{\delta})$, which gives
that $\eta_1  = O(\sqrt{\delta})$.

Unrolling the recursion, and using the fact that $\eta_1  = O(\sqrt{\delta})$,
we obtain
\begin{equation} \label{eqn:etam}
\eta_m = O(m \sqrt{\delta}) \max\left\{1, \left(\max_{i \in [k]} \|v_i\|_2\right)^{m-1} \right\} \;.
\end{equation}

We now write $\Phi_m$ as the composition of a projection $\mathrm{P}_m$ onto $U_m$
composed with an isometry $\mathrm{R}_m: U_m \to \R^k$.
We establish the following claim:

\begin{claim} \label{clm:M-proj}
We have that $\| M_m - \mathrm{P}_m(M_m) \|_2
\leq s \eqdef  \eta_m \, \sqrt{k} \, \left(\sum_{i=1}^k w_i \right)^{1/2}$.
\end{claim}
\begin{proof}
Since $M_m = \sum_{i=1}^k \sqrt{w_i} \, x_{i,m}$, by linearity of $\mathrm{P}_m$
it follows that
$\mathrm{P}_m(M_m) =  \sum_{i=1}^k \sqrt{w_i} \, \mathrm{P}_m (x_{i,m})$.
Therefore,
\begin{eqnarray*}
\| M_m - \mathrm{P}_m(M_m) \|_2
&=&  \sum_{i=1}^k \sqrt{w_i} \, \| x_{i,m}- \mathrm{P}_m (x_{i,m}) \|_2 \\
&\leq& \eta_m \, \sum_{i=1}^k \sqrt{w_i} \\
&\leq& \eta_m \, \sqrt{k} \, \left(\sum_{i=1}^k w_i \right)^{1/2} \;,
\end{eqnarray*}
where the first inequality follows from the definition of $\eta_m$
and the second is Cauchy-Schwarz.
\end{proof}

By construction, $A$ is the average of $N$ copies of $\mathrm{R}_m(\mathrm{P}_m(\mathcal{S}_m))$
on independent samples from the distribution $\mathcal{D}_m$.
By assumption, $\mathrm{P}_m(\mathcal{S}_m)$ has mean
$\mathrm{P}_m(M_m)$ and covariance bounded by $V$
(since $\mathrm{P}_m$ is a projection).
Given that these random variables lie in $U_m$, which is $k$-dimensional,
we have that with probability $1-\tau/4$ over samples from $\mathcal{D}_m$ it holds
$$\left\| Average(\mathrm{P}_m(\mathcal{S}_m)) - \mathrm{P}_m(M_m) \right\|_2
= O\left(\sqrt{\frac{k V}{N \tau}} \right) \;.$$
Since $\mathrm{R}_m$ is an isometry, we also get that
with probability $1-\tau/4$  over samples from $\mathcal{D}_m$
$$\|A - \mathrm{R}_m(\mathrm{P}_m(M_m))\|_2 = O\left(\sqrt{\frac{k V}{N \tau}} \right)  \;.$$

By similar logic, with probability $1-\tau/4$ over samples from $\mathcal{D}'$,
we have that
$$\E_{\nnew{\mathcal{X}} \sim X} [\|B(\nnew{\mathcal{X}})-\mathrm{R}_m(\mathrm{P}_m(T(\nnew{\mathcal{X}})))\|^2_2]^{1/2} =
 O\left(\sqrt{\frac{k V}{N \tau}} \right) \;.$$

Using the above, with probability at least $1-\tau$ over the samples drawn from $\mathcal{D}$
and $\mathcal{D'}$, we obtain the following chain of (in)equalities.
\begin{eqnarray*}
&& \langle A , B(\nnew{\mathcal{X}}) \rangle \\
&=& \langle \mathrm{R}_m(\mathrm{P}_m(T(\nnew{\mathcal{X}}))), \mathrm{R}_m(\mathrm{P}_m(M_m)) \rangle +
O\left(\sqrt{kV/(N \tau}) \right) \left(\|\mathrm{P}_m(T(\nnew{\mathcal{X}}))\|_2 + \|\mathrm{P}_m(M_m)\|_2\right) \\
&=& \langle \mathrm{P}_m(T(\nnew{\mathcal{X}})), \mathrm{P}_m(M_m) \rangle +
O\left(\sqrt{kV/(N \tau)} \right) \left( \|\mathrm{P}_m(T(\nnew{\mathcal{X}}))\|_2 + \|M_m\|_2 + s \right)
\quad \quad \textrm{($\mathrm{R}_m$ is an isometry)} \\
&=& \langle T(\nnew{\mathcal{X}}), \mathrm{P}_m(M_m) \rangle +
O\left(\sqrt{kV/(N \tau)}  \right) \left(\|\mathrm{P}_m(T(\nnew{\mathcal{X}}))\|_2 + \|M_m\|_2 + s\right) \quad \quad \textrm{($\mathrm{P}_m$ is a projection)} \\
&=& \langle T(\nnew{\mathcal{X}}), M_m \rangle +
O\left(\sqrt{kV/(N \tau)} \right) \left(\|\mathrm{P}_m(T(\nnew{\mathcal{X}}))\|_2 + \|M_m \|_2 + s\right)
+ s \, \| \proj_{M_m- \mathrm{P}_m(M_m)}(T(\nnew{\mathcal{X}})) \|_2 \;.
\end{eqnarray*}
Note that all of the computations before the last two steps of our algorithm
are independent of $T(\nnew{\mathcal{X}})$, and thus $T(\nnew{\mathcal{X}})$
is independent of $\mathrm{P}_m$
and the vector $M_m - \mathrm{P}_m(M_m)$. Since $\mathrm{P}_m(T(\nnew{\mathcal{X}}))$
is the projection of $T(\nnew{\mathcal{X}})$ onto $U_m$,
a $k$-dimensional subspace, we have that $\E_{\nnew{\mathcal{X}} \sim X}[\|\mathrm{P}_m(T(\nnew{\mathcal{X}}))\|^2_2]^{1/2} = O(\sqrt{Vk})$.
Since $M_m-\mathrm{P}_m(M_m)$ is just a vector, the square root of the expected squared error
of the projection of $T(\nnew{\mathcal{X}})$
onto that vector is $O(\sqrt{V})$.

By taking the expectation over $\nnew{\mathcal{X}} \sim X$ in the above,
using Claim~\ref{clm:M-proj}, \eqref{eqn:delta} and \eqref{eqn:etam}
completes the proof of Proposition~\ref{prop:main}. \qed

\section{Learning Theory Applications} \label{sec:apps}

The structure of this section is as follows: We start in Section~\ref{ssec:app-tech} with some technical machinery
that is useful in all our subsequent learning applications. Section~\ref{ssec:sum-relus} presents our results on one-hidden-layer
neural networks. Section~\ref{ssec:gmms} presents our algorithmic results for mixtures of spherical Gaussians (both
density and parameter estimation). Finally, Section~\ref{ssec:mlr} presents our algorithm for learning mixtures of linear regressions.

\subsection{Technical Machinery} \label{ssec:app-tech}

Before we get to the applications,
we require the following technical tool that we briefly motivate.
For our learning applications, we need to compute inner products of functions/distributions
with Hermite polynomials. We have technology for computing inner products of tensors
with objects that can be expressed as the expectation of sequential tensor computations.
Unfortunately, the Hermite tensor $H_n(x)$ cannot conveniently be written in this form.

\cite{LL21-opt} resolves this issue by expressing $H_n(x)$ as a sum of tensors of $x$'s and $I$'s
and replacing each copy of $I$ by $y_i \otimes y_i$ for independent Gaussians $y_i$
(noting that $\E[y_i \otimes y_i] = I$). We follow a different, more efficient approach: we take
all of the copies of $I$ in our product and replace them by
a tensor power of a single random Gaussian vector $y$.
This has the same expectation, but as a sum of only $2^n$
rank-$1$ terms --- rather than roughly $n^n$. Furthermore, we can compute this quantity more cleverly
using a sequential tensor computation of size only $O(n)$.

\begin{definition}[Extended Hermite Tensor] \label{def:Hxy}
For $n\in \N$ and $x, y \in \R^d$, we define
\[ H_n(x,y) = \frac{1}{\sqrt{n!}} \sum_{\substack{\text{Partitions of $[n]$ into $S_1, S_2$}\\ \text{with $|S_2|$ even}}} (-1)^{|S_2|/2} x^{\otimes_{S_1}} \otimes y^{\otimes_{S_2}} = \mathrm{Re}((x+iy)^{\otimes n})/\sqrt{n!} \;.\]
\end{definition}

Note that $H_n(x,y)$ is an order-$n$ tensor of dimension $d$.
The important property of the above definition is that
$H_n(x,y)$ can be computed by a sequential tensor computation of size $O(n)$
and that $\E_{Y \sim N(0,I)}[H_n(x,Y)] = H_n(x)$.
This follows from the fact that
\[ \E[Y^{\otimes 2t}] = \sum_{\text{partitions $P$ of $[2t]$ into sets of size $2$}} I^{P} \;.\]
Plugging this in to the expansion of $H_n(x,y)$ gives exactly the standard expansion of $H_n(x)$.

\begin{remark} \label{rem:comp}
{\em
All examples in~\cite{LL21-opt} involve explicitly writing the tensors as sums of rank-$1$ tensors
(i.e., products of vectors). While these are all STCs, not all STCs can be efficiently written in this way.
For example, if expanded out, $H_n(x,y)$ can be written as a sum of $2^n$ rank-$1$ tensors.
However, thinking of it as the real part of $(x+iy)^{\otimes n}$,
we can use a dynamic program keeping track of the real and imaginary parts,
and write it as the output of a STC of size $O(n)$ instead.
}
\end{remark}

To apply Proposition~\ref{prop:main},
we also need to bound the second moment of $H_n(X,Y)$ over $Y \sim N(0,I)$ and
independent $X \sim N(\mu,I)$.

\begin{lemma}[Second Moment Bound of $H_n(X,Y)$] \label{lem:cov-hxy}
We have that $\E_{Y, X}[H_n(X,Y)\otimes H_n(X,Y)]$, where $X \sim N(\mu, I)$, $Y \sim N(0, I)$
and $X, Y$ are independent, is bounded above by $O(\|\mu\|_2^2/n + 1)^n$.
\end{lemma}
\begin{proof}
The expectation of $H_n(X,Y) \otimes H_n(X,Y)$ can be expressed as follows:
\begin{align*}
\E_{X, Y}[H_n(X,Y) & \otimes H_n(X,Y)] = \\
& \frac{1}{n!} \sum_{S_1, S_2 \text{partition of } [n]} (-1)^{|S_2|/2} x^{\otimes S_1}y^{\otimes S_2} \otimes
\sum_{T_1, T_2 \textrm{ partition of } [n]} (-1)^{|T_2|/2} x^{\otimes T_1}y^{\otimes T_2} \;.
\end{align*}
Viewing the above as a partition of $[2n]$ into
$R_1 = S_1 \cup T_1$ and $R_2 = S_2 \cup T_2$, we can write
\[\E_{X, Y}[H_n(X,Y) \otimes H_n(X,Y)] =
\frac{1}{n!} \sum_{\substack{\text{Partitions of $[2n]$ into}\\ \text{$R_1$, $R_2$ with
$|R_2|, |R_2 \cap [n]|$ even}}}
(-1)^{|R_2|/2} x^{\otimes R_1} y^{\otimes R_2}\]

\[=  \frac{1}{n!} \sum_{\substack{\text{Partitions of $[2n]$ into}\\ \text{$R_1$, $R_2$ with
$|R_2|, |R_2 \cap [n]|$ even}}} (-1)^{|R_2|/2} \sum_{
\substack{\text{partitions $P_1$ of $R_1$}\\ \text{into sets of size 1 and 2},}
\substack{\text{partitions $P_2$ of $R_2$}\\ \text{into sets of size 2}}}
\otimes_{\{i\} \in P_1}\mu_i \otimes_{\{i,j\} \in P_1 \textrm{ or } P_2} I_{i,j} \;,\]
where the $\mu_i$ denotes a $\mu$ in the $i$-th tensor
slot and $I_{i,j}$ denotes a copy of the identity in the $i$th and $j$th.

If we now consider $P_1 \cup P_2$, we get a partition of $[2n]$ into sets of size $1$ and $2$
with the sets of size $2$ coming from either $P_1$ or $P_2$
(which we label as type $1$ and type $2$).
We have the restriction that the number of elements of $[n]$ in pairs of type $2$ is even.
Thus, we get
\[ \frac{1}{n!} \sum_{
\substack{\text{Partitions $P$ of $[2n]$ into sets of size 1 and 2}
\\ \text{with the sets of type 2 labelled 1 \& 2} \\ \text{and an even number of elements of $[n]$ in pairs of the latter type
}}}
(-1)^{\text{number of pairs of type 2}} \otimes_{\{i\} \in P} \mu_i \otimes_{\{i,j\}  \in P} I_{i,j} \;.\]

We think of this as first choosing the partition $P$ of $[2n]$ into sets of size $1$ and $2$
and only later choosing which to label as type $1$ and type $2$.
We note that if there is any pair contained entirely in $[n]$ or entirely in $[2n] \setminus [n]$,
then switching its type from $1$ to $2$ or back keeps the term the same but reverses its sign.
Thus, these terms cancel out.

Hence, we have:
\[  \frac{1}{n!} \sum_{\substack{\text{Partitions $P$ of $[2n]$ into sets of size 1 and 2} \\
\text{with an even number of sets of size 2} \\ \text{all of which cross from $[n]$ to $[2n]$
}}}
\otimes_{\{i\} \in P} \mu_i \otimes_{\{i,j\}  \in P} I_{i,j} \;.\]
Considering terms with exactly $k$ sets of size $2$, there are
$\binom{n}{k}^2$ ways to choose the sets of size $1$,
and $k!$ ways to choose how to pair up the remaining elements;
and at most $2^k$ ways to label the sets of size $2$.
When thought of as a $d^n \times d^n$ matrix, the tensor in the sum then
has spectral norm at most $\|\mu\|_2^{2n-2k}$. Thus,
the spectral norm of the covariance is at most
$$\frac{2^{O(n)}}{n^n} \; \sum_k \|\mu\|_2^{2n-2k} \sqrt{n}^{2k}
= O(\|\mu\|_2^2/n + 1)^n \;,$$
completing the proof.
\end{proof}

\begin{remark} \label{rem:GMM-cluster-opt}
{\em For the case of a mixture of $k$ identity covariance Gaussians,
the above second moment bound is polynomial in $k$ times
the largest mean, even when $n = O(\log(k))$. This bound is the key ingredient that
allows us to obtain the optimal mean separation condition of $O(\sqrt{\log(k)})$
for the parameter estimation problem,  improving on the bound of~\cite{LL21-opt}.}
\end{remark}

\subsection{Learning Positive Linear Combinations of Non-Linear Activations} \label{ssec:sum-relus}

\paragraph{Problem Setup}
We start by formally defining the target class of functions to be learned.

\begin{definition}[Positive Linear Combinations of an Activation Function] \label{def:one-hidden-nn}
Let $\sigma: \R\rightarrow\R$ be an activation function.
Denote by $\mathcal{C}_{\sigma,d,k}$ the class of functions on $\R^d$
given by a one-layer network with positive coefficients using the activation $\sigma$.
In particular, a function $F \in \mathcal{C}_{\sigma,d,k}$ if and only if there exist $k$ unit vectors
$v_i \in \R^d$ and non-negative coefficients $w_i \in \R_+$, $i \in [k]$,
with $\sum_{i=1}^k w_i=1$
such that $F(x) = \sum_{i=1}^k w_i \sigma(v_i \cdot x)$.
\end{definition}

This is a prototypical family of neural networks whose learnability has been extensively
studied over the past decade; see, e.g.,~\cite{SedghiJA16,
ZhongS0BD17, GeLM18, BakshiJW19, DKKZ20, DK20-ag, ChenKM21}.
A particularly noteworthy special case is that of ReLU networks,
where $\sigma(u) = \relu(u) := \max(0,u).$
As the majority of prior work on this problem, we will assume that the feature vectors
are normally distributed.
The definition of the PAC learning problem in our setting is the following.

\begin{definition}[PAC Learning $\mathcal{C}_{\sigma,d, k}$]\label{def:PAC-sums-ReLUs}
The PAC learning problem for the class $\mathcal{C}_{\sigma,d, k}$
is the following: The input is a multiset of i.i.d.\ labeled examples $(x, y)$,
where $x \sim N(0,I)$ and $y = F(x)$, for an unknown $F \in \mathcal{C}_{\sigma, d, k}$,
for a known activation $\sigma$. 
The goal of the learner is to output a hypothesis
$H: \R^d \to \R$ that with high probability is close to $F$ in $L_2$-norm, i.e.,
satisfies $\|H-F\|_2 \leq \eps$.
The hypothesis $H$ is allowed to lie in any efficiently
representable hypothesis class $\mathcal{H}$.
\end{definition}

Our main algorithmic result in this context is the following.

\begin{theorem}[Learning Algorithm for $\mathcal{C}_{\sigma, d, k}$]\label{thm:formal-sum-ReLUs}
Suppose that the activation $\sigma$ has $\|\sigma\|_4 = O(1)$.
Letting $c_{\sigma,t} := \E[\sigma(G)h_t(G)]$ be the $t^{th}$ Hermite coefficient of $\sigma$,
suppose additionally that for some $\eps,\delta>0$ and some positive $n$ it holds:
\begin{itemize}
\item $\sum_{t>n} c_{\sigma,t}^2 < \eps^2/4$
\item $|c_{\sigma,t}| \geq \delta$ for all $1\leq t \leq 2n$ unless $t$ is odd and $c_{\sigma,t}=0$.
\end{itemize}
Then there is a PAC learning algorithm for $\mathcal{C}_{\sigma, d, k}$ with respect to
the standard Gaussian distribution on $\R^d$ with the following performance guarantee:
Given $\eps>0$, $\sigma, k, d$ and access to labeled examples
from an unknown target $F \in \mathcal{C}_{\sigma,d, k}$,
the algorithm has sample and computational complexity
$\poly(dk/(\eps\delta)) 2^{O(n)}$,
and outputs an efficiently computable hypothesis $H: \R^d \to \R$ that with high probability satisfies
$\|H-F\|_2  \leq \eps $.
\end{theorem}

It is easy to see that the algorithm of Theorem~\ref{thm:formal-sum-ReLUs}
straightforwardly extends to the case that
the labels have been corrupted by random zero-mean additive noise.

As will become clear from the analysis that follows, our PAC learning algorithm is not proper.
The hypothesis $H$ it outputs is a succinct description of a low-degree polynomial
(namely, of degree $n$). Specifically, each Hermite coefficient will be expressed
in compressed form, as per Proposition~\ref{prop:main}.
This succinct description allows us to produce an efficient evaluation oracle for $H(x)$.

\begin{remark}
{\em Theorem \ref{thm:formal-sum-ReLUs} should provide useful results 
for \emph{most} activation functions $\sigma$ that are not odd 
(which would lead to vanishing of the even degree coefficients). 
Unfortunately, for many such activations, the analysis will be somewhat subtle---as one needs 
to show that $c_{\sigma,t}$ is not too close to $0$, which usually can only be done 
by computing these Hermite coefficients exactly. 
The two examples given below correspond to common activations 
for which this could be easily done.}
\end{remark}

In particular, for the special case of sums of ReLUs we obtain the following:
\begin{corollary} \label{cor:sum-relu}
For $\sigma(u) = \relu(u)$, there is a PAC learning algorithm that learns
$\mathcal{C}_{\sigma,d,k}$ to error $\eps$
with sample and computational complexity $\poly(d,k) 2^{\poly(1/\eps)}$.
\end{corollary}
\begin{proof}
Given Theorem \ref{thm:formal-sum-ReLUs}, this amounts to proving some
basic statements about the Fourier coefficients of $\sigma$.
In particular, it follows immediately from the claim that $c_{\sigma,t}=0$ if $t>1$ is odd and
$$c_{\sigma, t} = (-1/4)^{(t-2)/4}\sqrt{\binom{t-2}{(t-2)/2}}/\sqrt{2\pi t(t-1)} = \Theta(t^{-5/4})$$
if $t$ is even. This follows from Lemma 3.1 of~\cite{DK23-schur}. Thus, we may take $n = O(\eps^{-2/3})$ and the rest follows.
\end{proof}

Another application is to periodic activations 
given by $\sigma(u) = \cos(\gamma u)$ for some frequency parameter $\gamma$. 
While the majority of work on neural networks considers  
monotonic activations (e.g., ReLUs), an important branch of the literature
considers networks with periodic activations, specifically the cosine activation. 
Such activations are particularly useful in signal processing and computer vision applications,
where it has been empirically observed that networks with periodic activations are capable 
of representing details in the signals better than ReLU networks. The reader is referred to
~\cite{SitzmannMBLW20, MildenhallSTBRN20, RaiNMZYYFMPRABC21, VargasPHA24} and references therein.

In particular, we show:
\begin{corollary} \label{cor:sum-cos}
For $\sigma(u) = \cos(\gamma u)$, for some parameter $\gamma > 0$, 
there is a PAC learning algorithm that learns
$\mathcal{C}_{\sigma,d,k}$ to error $\eps$
with sample and computational complexity $2^{O(1/\gamma^2)}\poly(dk/\eps)$.
\end{corollary}

While the complexity of our algorithm blows up as $\gamma$ goes to infinity, 
there is evidence that this kind of dependence is necessary even for $k=1$. 
In particular,~\cite{DiakonikolasKRS23} (see Theorem 1.13) implies that any 
Statistical Query (SQ) algorithm learning a single cosine activation 
to non-trivial error requires $\exp(\min(\Omega(\gamma^2),d^{\Omega(1)}))$ resources. 
Furthermore, ~\cite{SongZB21} shows that if a tiny amount of adversarial label noise is added, 
then the problem is hard for polynomially-sized $\gamma$ under certain cryptographic assumptions.
Note that for $\gamma \ll 1/\sqrt{\log(dk/\eps)}$, our algorithm runs in 
$\poly(dk/\eps)$ time. By the above discussion, this is best possible for SQ algorithms. 

\begin{proof}[Proof of Corollary~\ref{cor:sum-cos}]
In light of Theorem \ref{thm:formal-sum-ReLUs}, 
this result comes down to calculating $c_{\sigma,t}$. 
Note that $c_{\sigma,t}$ is the real part of
$$
\frac{1}{\sqrt{2\pi}}\int_{-\infty}^\infty h_t(x) e^{i\gamma x} e^{-x^2/2}dx.
$$
In order to analyze this, we begin by computing
$$
\frac{1}{\sqrt{2\pi}}\int_{-\infty}^\infty h_t(x) e^{\alpha x} e^{-x^2/2}dx.
$$
for real valued $\alpha$. By completing the square, we see that this is
$$
e^{\alpha^2/2} \left( \frac{1}{\sqrt{2\pi}}\int_{-\infty}^\infty h_t(x)  e^{-(x-\alpha)^2/2}dx\right) = e^{\alpha^2/2} \E[h_t(G+\alpha)].
$$
Note that by Lemma 2.7 of~\cite{Kane20} 
(and noting the difference in normalization between their Hermite polynomials and ours), 
we have that
$$
\E[h_t(G+\alpha)] = \alpha^t/\sqrt{t!}.
$$
Thus, we have that for real $\alpha$
$$
\frac{1}{\sqrt{2\pi}}\int_{-\infty}^\infty h_t(x) e^{\alpha x} e^{-x^2/2}dx = e^{\alpha^2/2} \alpha^t/\sqrt{t!}.
$$
However, by analytic continuation, this must also hold for all complex $\alpha$ as well. 
Thus, plugging in $\alpha=i\gamma$ and taking the real part we have that
$$
c_{\sigma,t} = 
\begin{cases} e^{-\gamma^2/2} \gamma^t (-1)^{t/2} / \sqrt{t!} & 
\textrm{if }t \textrm{ is even} \\ 0 & 
\textrm{if }t\textrm{ is odd}. 
\end{cases}
$$

From this, it is not hard to see that it suffices to take $n = O(1/\gamma^2 + \log(1/\eps))$ and $\delta = \Omega(e^{-\gamma^2}\eps^2)$ in Theorem \ref{thm:formal-sum-ReLUs}, yielding our result.
\end{proof}

\begin{remark}
{\em Corollary~\ref{cor:sum-cos} easily generalizes  
to periodic activations $\sigma(u) = \cos(\gamma u + \theta)$ with an additional phase $\theta$, 
so long as $\theta$ is not close to a multiple of $\pi/2$
(which would again lead to vanishing Fourier coefficients).}
\end{remark}

\medskip

To prove Theorem~\ref{thm:formal-sum-ReLUs},
our basic strategy is to study the Hermite expansion of $F$,
which can be related to the kind of tensors approximated in Proposition \ref{prop:main}.
Our hypothesis will then be an approximation to the low degree part of the Hermite expansion of $F$.
To make this work, we will first need to understand the Hermite expansion of the relevant functions.

\paragraph{Hermite Decomposition of Functions in $\mathcal{C}_{\sigma, d, k}$}

We will require the following lemma:

\begin{lemma}[Hermite Expansion of $\mathcal{C}_{\sigma, d, k}$] \label{lem:Hermite-sum-ReLUs}
Let $F: \R^d \to \R$ be any function of the form
$F(x) = \sum_{i=1}^k w_i \sigma(v_i \cdot x)$ for some $v_i \in \R^d$ with $\|v_i\|_2=1$
and $w_i \in \R_+$. Then $F$ has Hermite expansion
$F(x) = \sum_{m=0}^{\infty} \langle T_m, H_m(x) \rangle$,
where $H_m$ is the normalized Hermite tensor (Definition~\ref{def:Hermite-tensor}) and
$T_m \in (\R^d)^{\otimes n}$ is defined by
\begin{equation} \label{eqn:Tn-sum-ReLUs}
T_m \eqdef c_{\sigma,m} \sum_{i=1}^k w_i v_i^{\otimes m}.
\end{equation}
\end{lemma}

\begin{proof}
These calculations have essentially appeared in prior work.
For the sake of completeness, here we show how to directly deduce the lemma from prior work.

By orthogonality of the Hermite tensors, we have that the tensor $T_n$
in the Hermite expansion of $F$, $F(x) = \sum_{m=0}^{\infty} \langle T_m, H_m(x) \rangle$,
is defined by
$T_m =  \E_{X \sim N(0, I)}[F(X) H_m(X)]$.

By definition, $\E_{G \sim N(0, 1)} [\sigma(G) h_m(G)] = c_{\sigma,m}$.
This in turn implies that for any unit vector $v \in \R^d$, it holds that
$\sigma(v \cdot x) = \sum_{m=0}^\infty c_{\sigma,m} \langle H_m(x), v^{\otimes m} \rangle$
(Lemma 3.2 of~\cite{DK23-schur}). Via orthogonality, we additionally
get that
$ \E_{X \sim N(0, I)}[\sigma(v\cdot X) H_m(X)] = c_{\sigma,m} v^{\otimes m}$ (Corollary 3.3 of~\cite{DK23-schur}).

Since $F(x) = \sum_{i=1}^k w_i \sigma(v_i \cdot x)$, by linearity of expectation it follows that
$$\E_{X \sim N(0, I)}[F(X) H_m(X)] =  \sum_{i=1}^k w_i \E_{X \sim N(0, I)}[\sigma(v_i\cdot X) H_m(X)]
 = \sum_{i=1}^k w_i c_{\sigma, m} v_i^{\otimes m}  =  c_{\sigma,m}  \sum_{i=1}^k w_i  v_i^{\otimes m} \;.$$
This completes the proof.
\end{proof}

\paragraph{Proof of Theorem~\ref{thm:formal-sum-ReLUs}}
The key properties that enable our algorithm are that the tensor $H_m(x, y)$
of Definition~\ref{def:Hxy} (i) can be computed by a sequential tensor computation of size $O(m)$,
(ii) it satisfies $\E_{Y \sim N(0,I)}[H_m(x,Y)] = H_m(x)$, for any $x \in \R^d$,
and (iii) it has bounded second moment (Lemma~\ref{lem:cov-hxy}).

To leverage these properties in the context of our learning problem, we note that
the second property implies that
\begin{equation}\label{eqn:Tn-STC}
T_m =  \E_{X \sim N(0, I)}[F(X) H_m(X)] = \E_{X \sim N(0, I), Y \sim N(0, I)}[F(X) H_m(X,Y)] \;,
\end{equation}
where $X$ and $Y$ are independent. Indeed, the first equality is the definition of $T_n$
and the second follows using property (ii) above (by linearity of expectation).

To apply Proposition \ref{prop:main}, we need to relate $T_m$ (which we can approximate)
to the standard moment tensors, $M_m$. Fortunately, it follows from the definitions
that $M_t = T_t / c_{\sigma,t}$. By assumption, for $t\leq 2n$ even,
we have $|c_{\sigma,t}| > \delta$ and so we can approximate $M_t$
to error $\eta$ by approximating $T_t$ to error $\eta/\delta$.
If $t$ is odd and $c_{\sigma,t}=0$, we will have no need to estimate $M_t$.
Otherwise, for $t$ odd we can approximate $M_m$ to error $\eta$
by approximating $T_t$ to error $\eta/\delta$ again.

In particular, we will want to use Proposition~\ref{prop:main} to efficiently approximate the inner product
$\langle T_m, H_m(x) \rangle$ for $x \sim N(0, I)$. Given that such an efficient
computation is possible, the hypothesis of our learning algorithm
will be an approximation to the low-degree Hermite expansion of $F$ of the function
\begin{equation} \label{eqn:hyp-sum-ReLUs}
\wt{H}(x) =  \sum_{m=0}^{n} \wt{\langle T_m, H_m(x) \rangle}.
\end{equation}
In particular, so long as our approximation to $\langle T_m, H_m(x) \rangle$
has $L_2$ error (over Gaussian $x$) at most $\eps/(4n)$ for each $m$, then we have that
\begin{align*}
\| H(x) - \wt{H}(x)\|_2 & \leq \sum_{m=0}^n \| {\langle T_m, H_m(x) \rangle} - \wt{\langle T_m, H_m(x) \rangle}\|_2
+ \left\| \sum_{m>n} \langle T_m, H_m(x) \rangle \right\|_2\\
& \leq \sum_{m=0}^n \eps/(4n) + \sqrt{\sum_{t>n} \|T_t\|_2^2 }\\
& \leq \eps/2 + \sqrt{\sum_{t>n} c_{\sigma,t}^2 \left(\sum_{i=1}^k w_i \|v_i^{\otimes t}\|_2 \right)^2}\\
& \leq \eps/2 + \sqrt{\sum_{t>n} c_{\sigma,t}^2}\\
& \leq \eps.
\end{align*}

In order to compute $\wh{H}$, we will need use to Proposition \ref{prop:main}
to approximate $\langle T_m, H_m(x) \rangle $ for all $m \leq n$
for which $c_{\sigma,m} \neq 0$. We do this by noting that
$\langle T_m, H_m(x) \rangle = c_{\sigma,m} \langle M_n, H_m(x)\rangle$.
In particular, for $t\leq 2m$ even or equal to $m$ we have that
$$
M_t = \E[F(x)H_t(x,y)/c_{\sigma,t}] \;,
$$
where $x$ and $y$ are independent Gaussians.
Given samples from $(x,F(x))$ and simulated Gaussians $y$,
one can compute $F(x)H_t(x,y)/c_{\sigma,t}$ as
a sequential tensor computation of order $t$ and size $O(t)$. Similarly, we have that
$$
H_m(x) = \E_y[H_t(x,y)].
$$
The second moment of $H_t(x,y)/c_{\sigma,t}$ is $O(1)^t$ by Lemma \ref{lem:cov-hxy}.
To bound the second moment of $F(x)H_t(x,y)/c_{\sigma,t}$, we note that it suffices to bound
$$
\E[(F(x) \langle A,H_t(x) \rangle/c_{\sigma,t})^2]
$$
for any $t$-tensor $A$ with $\|A\|_2 \leq 1$.
By Lemma \ref{lem:cov-hxy}, we have that
$$
\E[(\langle A,H_t(x) \rangle)^2] = O(1)^t.
$$
Since $\langle A,H_t(x)\rangle$ is a degree-$t$ polynomial of Gaussians,
by Gaussian hypercontractivity, we have that
$$
\E[(\langle A,H_t(x) \rangle)^4] = O(1)^t \;,
$$
and thus, by Holder's inequality, that
$$
\E[(F(x) \langle A,H_t(x) \rangle/c_{\sigma,t})^2]
\leq O(1)^t \|F(x)\|_4^2 / c_{\sigma,t}^2
\leq O(1)^t/\delta^2 \;.
$$
Thus, for $m\leq n$ and $c_{\sigma,m}\neq 0$,
we can apply Proposition \ref{prop:main} to get an approximation
$\wt{\langle T_m, H_m(x) \rangle}$ to ${\langle T_m, H_m(x) \rangle}$
with $L_2$ error $\eps/(4n)$ with sample and time complexity
$\poly(dk/(\eps\delta))2^{O(n)}$.
This completes our proof. \qed

\subsection{Learning Mixtures of Spherical Gaussians} \label{ssec:gmms}

In Section~\ref{ssec:gmm-dens},
we give our density estimation algorithm.
In Section~\ref{ssec:gmm-cluster}, we give our parameter estimation algorithm.

\paragraph{Problem Setup}
We start by formally defining the target distribution class to be learned.

\begin{definition}[Mixtures of Spherical Gaussians] \label{def:GMMs}
A $k$-mixture of spherical Gaussians is a distribution on $\R^d$
defined by $F = \sum_{i=1}^k w_i N(\mu_i,  I)$,
where $\mu_i \in \R^d$ are the unknown mean vectors and $w_i \geq 0$,
with $\sum_{i=1}^k w_i = 1$, are the mixing weights.
\end{definition}

\noindent We will consider both density estimation and parameter estimation.
In density estimation, we want to output a hypothesis distribution
whose total variation distance from the target distribution is small.
Recall that an \emph{$\eps$-sampler} for a distribution $D$ is a circuit $C$
that on input a set $z$ of uniformly random bits it generates
then  $y \sim D',$ for some distribution $D'$ which has
$\dtv(D',D) \leq \eps.$

\begin{definition}[Density Estimation for Mixtures of Spherical Gaussians]\label{def:learn-GMMs}
The density estimation problem for mixtures of spherical Gaussians
is the following: The input is a multiset of i.i.d.\ samples in $\R^d$ drawn from an
unknown $k$-mixture $F = \sum_{i=1}^k w_i N(\mu_i,  I)$.
The goal of the learner is to output a (sampler for a) hypothesis distribution
$H$ such that with high probability $\dtv(H, F) \leq \eps$.
\end{definition}

\begin{definition}[Parameter Estimation for Mixtures of Spherical Gaussians]\label{def:cluster-GMMs}
The parameter estimation problem for mixtures of spherical Gaussians
is the following: The input is a multiset of i.i.d.\ samples in $\R^d$ drawn from an
unknown $k$-mixture $F = \sum_{i=1}^k w_i N(\mu_i,  I)$,
where the component means $\mu_i$ satisfy a pairwise separation condition.
The goal of the learner is to accurately estimate the weights and mean vectors of the components.
\end{definition}

\subsubsection{Density Estimation for Mixtures of Spherical Gaussians} \label{ssec:gmm-dens}

We consider mixtures of $k$ identity covariance Gaussians on $\R^d$
with the additional restriction that the mean vectors of the components
lie in a ball of radius $O(\sqrt{\log(k)})$. By subtracting the mean of the mixture,
it suffices to consider the case that each component mean has magnitude
$O(\sqrt{\log(k)})$.
Our main algorithmic result in this context is the following.

\begin{theorem}[Density Estimation Algorithm for Mixtures of Spherical Gaussians with Bounded Means]\label{thm:formal-dens-GMMs}
There is an algorithm that given $\eps>0$ and $n=\poly(d, k, 1/\eps)$ samples
from an unknown $F= \sum_{i=1}^k w_i N(\mu_i,  I)$
on $\R^d$ with means of magnitude
$O(\sqrt{\log(k)})$, it runs in $\poly(n, d)$ time and outputs a (sampler for a)
hypothesis distribution $H$ such that with high probability $\dtv(H, F) \leq \eps$.
\end{theorem}

As for our previous application, our learning algorithm is not proper.
The hypothesis distribution $H$ will be such that
$H/G$, where $G$ is the pdf of $N(0, I)$,
will be a low-degree polynomial in compressed form,
as per Proposition~\ref{prop:main}, that allows for efficient sampling.

\medskip

To apply Proposition~\ref{prop:main} for this learning problem, we will need
the following basic facts on the Hermite decomposition of the relevant distributions.

\paragraph{Hermite Decomposition of Spherical Mixtures}

We will require the following lemma:

\begin{lemma}[Hermite Expansion of Spherical Mixtures] \label{lem:Hermite-GMMs}
Let $F: \R^d \to \R$ be any distribution of the form
$F(x) = \sum_{i=1}^k w_i N(\mu_i, I)$ for some $\mu_i \in \R^d$
and $w_i \in \R_+$ with $\sum_{i=1}^k w_i=1$. Let $G(x)$ be the pdf of the standard Gaussian
$N(0, I)$. Then $(F/G)(x)$ has Hermite expansion
$(F/G)(x) = \sum_{n=0}^{\infty} \langle T_n, H_n(x) \rangle$,
where $H_n$ is the normalized Hermite tensor (Definition~\ref{def:Hermite-tensor}) and
$T_n \in (\R^d)^{\otimes n}$ is defined by
\begin{equation} \label{eqn:Tn-GMMs}
T_n \eqdef (1/\sqrt{n!}) \sum_{i=1}^k w_i \mu_i^{\otimes n} \;.
\end{equation}
\end{lemma}

\begin{proof}
Since $(F/G)$ is in $L_2$ with respect to the Gaussian measure $G$,
a Hermite expansion does exist in the form of
$(F/G)(x) = \sum_{n=0}^{\infty} \langle T_n, H_n(x) \rangle $, where
$$T_n = \E_{x \sim G} [(F/G)(x) H_n(x)] = \E_{x \sim F} [H_n(x)] \;.$$
Since $F$ is a mixture, we have that
$\E[H_n(F)] = \sum_i w_i \E[H_n(N(\mu_i,I))]$.
By Lemma 2.7 of~\cite{Kane20}, this is equal to
$(1/\sqrt{n!}) \sum w_i \mu_i^{\otimes n}$.
This completes the proof.
\end{proof}

\paragraph{Proof of Theorem~\ref{thm:formal-dens-GMMs}.}
The proof is analogous to the proof of Theorem~\ref{thm:formal-sum-ReLUs}
for our previous application. We similarly leverage the fact that the tensor $H_n(x, y)$
of Definition~\ref{def:Hxy} (i) can be computed by a sequential tensor computation of size $O(n)$,
(ii) it satisfies $\E_{Y \sim N(0,I)}[H_n(x,Y)] = H_n(x)$, for any $x \in \R^d$,
and (iii) it has bounded second moment (Lemma~\ref{lem:cov-hxy}).

We start by noting that
the second property implies that the tensor $T_n$ of \eqref{eqn:Tn-GMMs} satisfies
\begin{equation}\label{eqn:Tn-STC-GMMs}
T_n =  \E_{X \sim F}[H_n(X)] = \E_{X \sim F, Y \sim N(0, I)}[H_n(X, Y)]\;,
\end{equation}
where $X$ and $Y$ are independent. Indeed, the first equality is the definition of $T_n$
and the second follows using property (ii) above (by linearity of expectation).

We will apply Proposition~\ref{prop:main} to efficiently approximate $T_n$.
Specifically, we will show how to efficiently approximate the inner product
$\langle T_n, H_n(x) \rangle$ for $x \sim N(0, I)$. Given that such an efficient
approximation is possible, we will show how to output an efficiently samplable
hypothesis distribution.

The first step will be to produce an evaluation circuit which computes an approximation to
$(F/G)(x)$ for Gaussian random $x$. We will achieve this via
an application of Proposition~\ref{prop:main}.

Note that
$$\|T_n\|_2^2 = O\left(\max_{i \in [k]} \|\mu_i\|_2^2/n\right)^{n/2}  = O(\log(k)/n)^{n/2} \;,$$
where the first equality follows by a direct calculation and the second follows by our assumed
bound on the $\|\mu_i\|_2$'s.

Therefore, for $n_0$ a sufficiently large constant multiple of $\log(k/\eps)$,
truncating the Hermite expansion of $F/G$ at $n_0$
introduces $L_2$-error (with respect to $G$) at most $\eps/2$, i.e.,
$$(\wt{F}/G)(x) = \sum_{n=0}^{n_0} \langle T_n, H_n(x) \rangle $$
is $\eps/2$-close to $F/G$ in $L_2$ norm. In particular, if we let
$\wt{F}(x) = G(x) \sum_{n=0}^{n_0} \langle T_n,H_n(x) \rangle$,
we have that the $L_1$-error between $F$ and $\wt{F}$ is less than $\eps/2$.

Using Proposition~\ref{prop:main}, we can come up with an evaluation circuit
which approximates $(\wt{F}/G)$ at random $x$ to small error.
The parameter $m$ in the statement of the proposition will be set to $m:=n$, for $n \leq n_0$.
The tensor $M_t$  will be set to $M_t: =\sum_{i=1}^k w_i v_i^{\otimes t}$, where
$w_i$ are the mixture weights and $v_i = \mu_i/\sqrt{\log(k)}$, where $\mu_i$ are
the component means of $F(x)$. By our assumption on the magnitude of the
$\mu_i$'s,  it follows that each $v_i$ has $\ell_2$-norm $O(1)$.
This re-parameterization is necessary in order to get the right error terms in our analysis.
We then have that
$$(\wt{F}/G)(x) = \sum_{n=0}^{n_0} \langle \sqrt{\log(k)^n/n!} H_n(x), M_n \rangle \;.$$
Combining the above with \eqref{eqn:Tn-STC-GMMs} gives that
$M_n = \sqrt{n!/\log(k)^n} \E[H_n(F,G)]$.

Given the above, the sequential tensor computation $\mathcal{S}_n$ has two vector inputs $v, u \in \R^d$ and
is defined so that $\mathcal{S}_n(u, v) = \sqrt{n!/\log(k)^n} H_n(u, v)$.
We have already argued that $H_n(u, v)$ can be computed by a sequential tensor computation of size $O(n)$.
The corresponding input distribution is the joint distribution $\mathcal{D}: = (X, Y)$,
where $X \sim F$ and $Y \sim N(0, I)$ are independent.
Since we have sample access to $F$ the distribution $\mathcal{D}$ is efficiently samplable.
Moreover, we have
$$\E_{(u, v,) \sim \D} [\mathcal{S}_n(u, v)] = \sqrt{n!/\log(k)^n}
\E_{X \sim F, Y \sim N(0, I)}[H_n(X,Y)] = M_n  \;.$$

We can use Lemma~\ref{lem:cov-hxy} to bound the covariance of
$\mathcal{S}_t$. We show the following:

\begin{claim} \label{clm:cov-St-GMMs}
We have that $\Cov_{(u, v) \sim \D} [\mathcal{S}_t(u, v)]$ is bounded above by
$V: = \poly(k/\eps)$.
\end{claim}
\begin{proof}
Using the fact that $F = \sum_{i=1}^k w_i N(\mu_i,I)$,
we find that the covariance of $H_n(F,G)$ is
$\sum_{i=1}^k w_i \Cov[H_n(N(\mu_i,I),G)] + \Cov[X]$,
where $X$ is the random variable that has expected value
$\E[H_n(N(\mu_i,I),G)] = \mu_i^{\otimes n}/\sqrt{n!}$
with probability $w_i$. By Lemma~\ref{lem:cov-hxy},
the first term above has operator norm at most $O(\log(k)/n+1)^n$
and the second term is similarly bounded.
For $n = O(\log(k/\eps))$, this is at most $\poly(k/\eps)$.
\end{proof}

The sequential tensor computation $\mathcal{F}_n$ and the distribution $\mathcal{D'}$
will be the same as in our previous application, up to rescaling.
Namely, $\mathcal{F}_n(u, v) = \sqrt{\log(k)^n/n!} H_n(u, v)$ and
$\mathcal{D}'$ is the distribution of $Y \sim N(0, I)$.
Since $H_n(X) = \E_{Y}[H_n(X,Y)]$, for $X \sim N(0, I)$,
and the second moment of $H_n(X, Y)$ is bounded above
by Lemma~\ref{lem:cov-hxy}, it follows that the second moment of
$\mathcal{F}_n$ is bounded above by $O(\log(k)/n)^n$.

Therefore, applying Proposition~\ref{prop:main},
we have an algorithm that runs in time $\poly(k, d, n_0, 1/\eps, 1/\tau)$
which with probability $1-\tau$ produces an approximation to
$\langle \sqrt{\log(k)^n/n!}H_n(x), M_n\rangle$, for $n \leq n_0$,
so that the $L_2$ (and thus $L_1$) expected error for $x \sim G$
is at most $\eps/(2(n_0+1))$. Summing this over $0 \leq n \leq n_0$
and taking $\tau < 1/(100 n_0)$, gives a 99\% probability of producing an approximation
to $(\wt{F}/G)$ with $L_1$ error at most $\eps/2$.
This in turn gives an approximation to $F/G$ with $L_1$ error at most $\eps$.

We next need to address how to go from an approximate evaluation oracle for $F/G$
to an approximate sampler for $F$. We achieve this via rejection sampling.
In particular, we produce a random sample $x \sim G$ and compute our oracle $R$ at $x$.
We then return $x$ with probability proportional to $R(x)$, and otherwise resample and repeat.
Note that if $R$ returned exactly $(F/G)$ and if we could keep $x$
with probability exactly proportional to $R(x)$, this would give $F$ exactly.

The first issue we have to face is that $R(x)$ is neither guaranteed
to be non-negative nor bounded,
and thus keeping $x$ with probability proportional to $R(x)$ is impossible.
To fix this, we need to truncate $R$. In particular, we define $R'$ so that
$R' = 0$ if $R < 0$; $R = A$ if $R > A$, and $R$ otherwise,
for $A$ some parameter we will choose shortly.

We claim that $\|R'-(F/G)\|_{1,G} := \E_{x \sim G}[|R(x) - (F/G)(x)|]$
is at most $3\eps$.

First, note that $\|R-(F/G)\|_{1,G} < \eps$.
Next, since $(F/G)$ is non-negative,
replacing $R(x)$ by $0$ when $R(x) < 0$, will only decrease the distance.
It remains to show that
$$\E_{x \sim G}[\max(R(x)-A, 0)] < 2\eps \;.$$
Since $\|R-(F/G)\|_{1,G} < \eps$, it suffices to show that
$\E_{x \sim G}[\max( (F/G)(x)-A, 0)] < \eps$.
The above is equal to
$\int_{t>A} \Pr[(F/G)(x) > t] dt$.
Using the definition of $F$, it follows that
$$(F/G)(x) \leq \max_i \frac{N(\mu_i,I)}{N(0,I)}(x) \leq \max_i \exp(\mu_i \cdot x) \;.$$
Thus, by a union bound
$$\Pr[(F/G)(x) > t] \leq \sum_i \Pr[\mu_i \cdot x > \log(t)] \leq
O(k) \exp\left(-\Omega\left(\log(t)/\sqrt{\log(k)}\right)^2 \right) \;.$$
Thus, if $A$ is a large constant degree polynomial in $k/\eps$,
this resulting integral is at most $\eps$.

If the above holds, then our algorithm can evaluate a function $R'(x)$
so that with 99\% probability $\|R-(F/G)\|_1 < 3 \eps$
(we can re-parameterize, by setting $\eps/6$ instead of $\eps$,
so this is actually less than $\eps/2$);
and $R'(x) \in [0,A]$ for all $x$.

We now use rejection sampling. We repeatedly sample $x \sim G$
and with probability $R'(x)/A$, we return $x$ and otherwise repeat this process.
Note that since $\|R'-F/G\|_1$ is small, $\E[R'(x)] > 1/2$,
and so each iteration has an $\Omega(1/A)$ chance of terminating.
Thus, in expectation, this procedure terminates after only $O(A)$
rounds with 99\% probability. Furthermore, the resulting distribution
that we sample from is proportional to $R' G$.
Observe that this is $\eps$-close to $(F/G) \, G = F$.
Thus, we are sampling from a distribution $\eps$-close to $F$, as desired.

Overall, this gives a density estimation algorithm
with complexity $\poly(d, k, 1/\eps)$,
completing the proof of Theorem~\ref{thm:formal-dens-GMMs}.
\qed

\subsubsection{Parameter Estimation for Mixtures of $O(\sqrt{\log(k)})$-Separated Gaussians} \label{ssec:gmm-cluster}

We start by recalling that \cite{LL21-opt}'s result on clustering/parameter estimation
for mixtures of spherical Gaussians only works
under pairwise mean separation of $\log(k)^{1/2+c}$,
for some constant $c>0$\footnote{For the sake of simplicity, this
description focuses on the case of uniform mixtures.}.
This is because, due to the way that they were approximating Hermite polynomials,
they could only approximate tensors of order $\log(k)/\log\log(k)$ in polynomial time.
Our extended Hermite tensors (Definition~\ref{def:Hxy} and Lemma~\ref{lem:cov-hxy})
improve upon this, allowing us to work with tensors of order $\log(k)$, and thus
obtain optimal (up to constant factors) $O(\sqrt{\log(k)})$ separation.

Specifically, we show the following statement that also handles general weights:

\begin{theorem}[Parameter Estimation for Mixtures of Spherical Gaussians] \label{thm:formal-cluster-GMMs}
Let $F = \sum_{i=1}^k w_i N(\mu_i,I)$ be a mixture of Gaussians
with $w_{\min} \leq \min w_i$ and let $\alpha > 0$ be an accuracy parameter.
Suppose that $s := \min_{i\neq j}\|\mu_i - \mu_j\|_2$ is at least a sufficiently large constant multiple of
$\sqrt{\log(1/(\alpha w_{\min}))}$ and furthermore that
$\max_{i\neq j} \|\mu_i - \mu_j\|_2 = O(\min_{i\neq j} \|\mu_i - \mu_j\|_2).$
Then there exists an algorithm that given $k, w_{\min}, \alpha, s$
and $\poly(d/(w_{\min}\alpha))$ i.i.d.\ samples from $F$,
runs in $\poly(n, d)$ time, and
outputs estimates $\tilde \mu_i$ and $\tilde w_i$ of $\mu_i$ and $w_i$,
so that with probability $2/3$, for some permutation $\pi$ of $[k]$
$$
|\tilde w_i - w_i|, \|\tilde \mu_i - \mu_i\|_2 \leq \alpha
$$
for all $i\in [k]$.
\end{theorem}

Note that \cite{LL21-opt} achieves this goal only for separation $\log(1/(\alpha w_{\min}))^{1/2+c}$.
While they do not require the condition that the largest pairwise distance is comparable to the smallest,
in order to remove this condition, they developed and leveraged a complicated recursive clustering
argument (see Sections 10 and 11 of \cite{LL21-opt}), which we expect can also be applied to our setting.

\begin{proof}[{\bf Proof of Theorem~\ref{thm:formal-cluster-GMMs}}]
Firstly, we note that by standard techniques,
we can reduce to the case of $d=k$; see, e.g.,~\cite{VempalaWang:02, DKS18-list}.
In particular, we may assume that $\max \|\mu_i-\mu_j\|_2 < \log(k)$
(or otherwise \cite{LL21-opt}'s result will apply anyway).
Given this and $\poly(d/(w_{\min}\alpha))$ samples,
we can approximate the covariance matrix of $F$ to inverse polynomial accuracy.
Letting $V$ be the span of the $k$ principal eigenvectors of this covariance matrix,
it is not hard to show that all of the $\mu_i$
lie within distance $\alpha/2$ of some translate of $V$
(this translate can be estimated by approximating the mean of the projection of $F$ onto $V^\perp$).
Thus, it suffices to consider our algorithm on the projection of $F$ onto $V$,
which is a $k$-dimensional subspace.

We next draw $\poly(k/(w_{\min}\alpha))$ samples from $F$.
We claim that with high probability we can cluster these samples
so that the samples drawn from the component $N(\mu_i,I)$ lie exactly in their own cluster.
If we can do this, it is straightforward to show that letting $\tilde \mu_i$
be the mean of the elements in the $i^{th}$ cluster and letting $\tilde w_i$
be the fraction of elements in the $i^{th}$ cluster will suffice for our estimates.
Following~\cite{LL21-opt}, in order to do this clustering,
we need only the following procedure:

Given $x,x' \sim F$ and $\poly(1/(w_{\min} \tau))$ i.i.d.\ samples from
$F$ (for some $\tau =\poly(k/(w_{\min}\alpha)) > 0$),
determine with probability at most $1-\tau$ of error whether or not $x$ and $x'$
were drawn from the same component.

For this, we consider the distribution $X'$ obtained by taking $(y-y')/\sqrt{2}$,
where $y,y'$ are independent samples from $F$.
We note that $X'$ is a mixture of at most $k^2$ Gaussians, namely
$$
X' \sim (w_1^2 + w_2^2 + \ldots + w_k^2) N(0,I) + \sum_{i\neq j} w_i w_j N((\mu_i-\mu_j)/\sqrt{2},I).
$$
We define $M_t = \sum_{i,j} w_i w_j (\mu_i - \mu_j)^{\otimes t}$ and note that
\begin{equation}\label{GMM cluster M est eqn}
M_t = \sqrt{t!}\E_{x \sim X', y \sim N(0,I)}[H_t(x,y)] \;.
\end{equation}
Our goal will be to use Proposition \ref{prop:main}
to approximate $P_t := \langle H_t((x-x')/\sqrt{2}), M_t\rangle$,
for $t$ an even integer bigger than a suitably large constant multiple of $\log(1/w_{\min})$.
In particular, we have that
$$
P_t = \sum_{i \neq j} w_i w_j \|\mu_i - \mu_j\|_2^t h_t(v_{i,j} \cdot (x-x')/\sqrt{2}) \;,
$$
where $v_{i,j}$ is the unit vector in the direction of $\mu_i - \mu_j$.

If $x$ comes from the $i^{th}$ component and $x'$ from the $j^{th}$,
$v_{i,j}\cdot(x-x')/\sqrt{2}$ is distributed as $N(\|\mu_i-\mu_j\|_2/\sqrt{2},1)$,
and thus with high probability is at least $s/2$. By direct computation, we have that
$$
h_t(x) = \frac{1}{\sqrt{t!}} \sum_{a=0}^{t/2} (-1)^a x^{t-2a}\binom{t}{2,2,2,\ldots,2,t-2a} = \frac{x^t}{\sqrt{t!}}(1+O((t/x)^2 + (t/x)^4 + \ldots + (t/x)^{t}).
$$
Therefore, if $|x|\leq t$, we have that $h_t(x) = O(t)^{t/2}$;
but if $|x|$ is at least a sufficiently large constant multiple of $t$,
we have that $h_t(x) \geq \frac{x^t}{2\sqrt{t!}}$.
In particular, $h_t(v_{i,j}\cdot(x-x')/\sqrt{2})$ will be at least $\Omega(s^2/\sqrt{t})^t$.
Since all of the other $h_t(v_{i',j'}\cdot (x-x')/\sqrt{2})$ are bounded below,
this implies that $P_{t} = w_{\min}^2 \Omega(s^2/\sqrt{t})^t.$

On the other hand if $x$ and $x'$ come from the same component,
$v_{i,j}(x-x')/\sqrt{2} \sim N(0,1)$, and with probability $1-\tau/3$ all of these
quantities have magnitude $O(\sqrt{\log(k/\tau)})$.
Direct computation similarly gives us that
$$
h_t(x) = \frac{1}{\sqrt{t!}} \sum_{a=0}^{t/2} (-1)^a x^{t-2a} \frac{t!}{(t-2a)!a!2^a} \;.
$$
If $|x| = O(\sqrt{t})$, the largest term above will occur
when $x^2 a \approx (t-2a)^2$, which occurs when $a$ is a constant multiple of $t$.
In this case, it is easy to see that the term in question is $O(1)^t$.
Thus, when $x$ and $x'$ come from the same component,
$P_{t} = O(s^2/t)^t$ (with a small enough implied constant in the big-$O$)
with high probability.

Therefore, by estimating $P_{t}$ to error $O(s)^t$,
we can reliably distinguish between the cases where $x$ and $x'$
come from the same component and the one where they do not.
This can be done directly using Proposition \ref{prop:main}
along with Equation \eqref{GMM cluster M est eqn} and the fact that
$H_t((x-x')/\sqrt{2}) = \E_{y\sim N(0,I)}[H_n((x-x')/\sqrt{2},y)].$
Applying Lemma \ref{lem:cov-hxy}, we find that both of these estimators
have second moment bounded by $O(1+s^2/t)^t$.
This can be estimated to appropriate error with $\poly(k/(w_{\min}\tau))$ samples.
\end{proof}

\subsection{Learning Mixtures of Linear Regressions} \label{ssec:mlr}

\textbf{Problem Setup} We start by defining the underlying probabilistic model.
A linear regression problem produces a distribution on $\R^{d+1}$. In particular,
given $\sigma>0$ and $\beta \in \R^d$, we define a distribution $L_{\beta,\sigma}$ on $\R^{d+1}$
given as the distribution of $(X,y)$, where $X\in \R^d$ is distributed as $N(0,I)$ and $y = \beta\cdot X + N(0,\sigma^2)$.
Notice that as a distribution over $\R^{d+1}$, $L_{\beta,\sigma}$ is just a Gaussian, namely
$$
L_{\beta,\sigma} \sim N\left( 0 , \left[ \begin{matrix} I & \beta \\ \beta^T & \|\beta\|_2^2 + \sigma^2 \end{matrix} \right] \right).
$$

\noindent We can now define the corresponding mixture model.

\begin{definition}[Mixtures of Linear Regressions] \label{def:mlr}
A $k$-mixture of linear regressions ($k$-MLR) with error $\sigma>0$,
is any distribution of the form $F = \sum_{i=1}^k w_i L_{\beta_i,\sigma}$
for some unknown mixing weights $w_i \geq 0$ with $\sum_{i=1}^k w_i =1$
and unknown vectors $\beta_i$. We will assume that the $\beta_i$'s
have some known upper bound $B$ on their norms,
namely $\|\beta_i\|_2 \leq B$ for $1\leq i \leq k$.
\end{definition}

\begin{definition}[Density Estimation for Mixtures of Linear Regressions]\label{def:learn-MLR}
The density estimation problem for mixtures of linear regressions
is the following: The input is a multiset of i.i.d.\ samples in $\R^{d+1}$ drawn from an
unknown $k$-mixture $F = \sum_{i=1}^k w_i L_{\beta_i,\sigma}$,
where $\sigma, k,$ and $ B$ are known.
The goal of the learner is to output a (sampler for a) hypothesis distribution
$H$ such that with high probability $\dtv(H, F) \leq \eps$.
\end{definition}

Our main algorithmic result for this problem is the following:

\begin{theorem}[Density Estimation Algorithm for $k$-MLR]\label{MLRThm}
Suppose that we are given sample access to a $k$-MLR distribution $F$ with $B,\sigma \leq 1$.
Then there exists an algorithm that given $k,\sigma,$ and $\eps>0$,
draws $N=\poly(k,d)\eps^{O(\sigma^{-2})}$ samples from $F$,
runs in $\poly(N, d)$ time, and returns a sampler for a distribution that is $\eps$-close to $F$ in total variation distance.
\end{theorem}

Note that the requirement that $\sigma,B\leq 1$ can be removed at the cost
of making the runtime $\poly(k,d)(1/\eps)^{1+O((B/\sigma)^2)}$.
This is done by computing an upper bound $A$ on $B+\sigma$ (for example by taking
many samples and returning something proportional to the largest value of $|y|$),
and renormalizing by replacing $y$ by $y/A$.


\paragraph{Proof of Theorem~\ref{MLRThm}.}
By dividing the $y$-values by $2$, we get a new distribution of the same form
with $\sigma$ and $\beta_i$ all half as big. Thus, we can reduce to the case
where $\sigma,B\leq 1/2$, which we assume below.

As in the Gaussian mixtures application, we will produce this distribution
through rejection sampling. In particular, we will produce a random sample $x$
from the standard Gaussian $G$, and then we will try to accept it with
probability proportional to $(F/G)(x)$. This depends on being able
to approximate the function $(F/G)(x)$ with small $L_1$ error.
We do this by trying to compute its Taylor expansion using Proposition \ref{prop:main}.

In the lemma below, we start by defining the appropriate parameter
moments for this setting.

\begin{lemma}\label{MLRMomentCompLem}
Writing $F=(X,y)$ with $X\in \R^d$ and $y\in \R$, we have that
\begin{equation} \label{MLGMomentTensorComputationEquation}
M_m := \sum_{i=1}^k w_i \beta_i^{\otimes m} = \E_{(X,y)\sim F, Y\sim N(0,I)}\left[(y^m/\sqrt{m!})H_m(X,Y) \right].
\end{equation}
Furthermore, if $B,\sigma \leq 1$, the second moment of $(y^m/\sqrt{m!})H_n(X,Y)$ above in any direction is bounded by
$$
V = 2^{O(m)}.
$$
\end{lemma}
\begin{proof}
To begin with, we note that the expectation of $Y$
of the right hand side of Equation \eqref{MLGMomentTensorComputationEquation}
is $\E_{(X,y)\sim F}[(y^m/\sqrt{m!})H_m(X)].$ Since $F$ is a mixture of the $L_{\beta_i,\sigma}$, this is
$$
\sum_{i=1}^k w_i \E_{(X,y)\sim L_{\beta_i,\sigma}}[(y^m/\sqrt{m!})H_m(X)].
$$
For the $i^{th}$ term of this sum, we note that $y = \beta_i\cdot X + \xi$,
where $\xi \sim N(0,\sigma^2)$ is an independent Gaussian.
Thus, $y^m/\sqrt{m!} = (\beta_i\cdot X+\xi)^m/\sqrt{m!}$.
For any fixed value of $\xi$, this means that $(y^m/\sqrt{m!})$ is
$\langle \beta^{\otimes m}, H_{m}(X)\rangle$
plus a polynomial of degree less than $m$ in $X$.
As $H_m(X)$ is orthogonal to polynomials of degree less than $m$, the above expectation is
$$
\sum_{i=1}^k w_i \E_X[\langle \beta_i^{\otimes m}, H_m(X) \rangle, H_m(x)] = M_{m} \;,
$$
by the fact that $\beta_i^{\otimes m}$ is symmetric
and the orthonormality properties of the Hermite tensors $H_m$.

For the second moment bound, we let $T$ be a tensor with $\|T\|_2 \leq 1$ and we wish to bound
$$
\E_{(X,y)\sim F, Y \sim N(0,I)}[((y^m/\sqrt{m!})\langle T, H_m(X,Y)\rangle)^2].
$$
We can do this by bounding the expectation on each component of the mixture.
We note that with $(X,y)\sim L_{\beta_i}$ and $Y\sim N(0,I)$,
all of the terms in the above expectation are polynomials of a Gaussian input.
By Lemma \ref{lem:cov-hxy}, we have that $\|\langle T, H_m(X,Y)\rangle\|_2 = O(1)^m$.
Similarly, a direct computation shows that $\|(y^m/\sqrt{m!})\|_2 = O(1)^m$.
Applying Gaussian Hypercontractivity (Fact~\ref{thm:hc})
and Holder's Inequality, we conclude that $\|(y^m/\sqrt{m!})\langle T, H_m(X,Y)\rangle \|_2 = O(1)^m$,
as desired.
\end{proof}

We next need to determine how to approximate $(F/G)(x)$.
Using Hermite analysis, we know that it is given by
$$
(F/G)(x) = \sum_{n=0}^\infty \langle T_n , H_n(x)\rangle \;,
$$
where
$$
T_n = \E[H_n(F)] = \sum_{i=1}^k w_i \E[H_n(L_{\beta_i,\sigma})].
$$
As $L_{\beta,\sigma}$ is just a Gaussian, we have by Lemma 2.7 of~\cite{Kane20}
(noting the difference in normalization between their Hermite polynomials and ours),
we have that $\E[H_n(L_{\beta,\sigma})]$ is $0$ if $n$ is odd and is
$$
\frac{(n-1)!!}{\sqrt{n!}}\mathrm{Sym}\left(\left[\begin{matrix} 0 & \beta \\
\beta^T & \|\beta\|_2^2 + \sigma^2 -1 \end{matrix} \right]^{\otimes n/2} \right)
$$
if $n$ is even. In particular, letting $e_y$ be the unit vector in the $y$-direction, this is
\begin{align*}
&  \frac{(n-1)!!}{\sqrt{n!}}\mathrm{Sym}\left(\left(2 \beta \otimes e_y + \|\beta\|_2^2 e_y \otimes e_y + (\sigma^2-1) e_y \otimes e_y  \right)^{\otimes n/2} \right)\\
= &  \frac{(n-1)!!}{\sqrt{n!}}\sum_{a+b+c=n/2} \binom{n/2}{a,b,c}2^a\|\beta\|_2^{2b} (\sigma^2-1)^c\mathrm{Sym}\left( \beta^{\otimes a}\otimes e_y^{\otimes n-a} \right) \;.
\end{align*}
Therefore, we have that
\begin{align*}
\langle T_n, H_n(X) \rangle & = \sum_{i=1}^k w_i \langle \E[H_n(L_{\beta_i,\sigma})], H_n(X) \rangle\\
& = \frac{(n-1)!!}{\sqrt{n!}} \sum_{i=1}^k w_i\sum_{a+b+c=n/2} \binom{n/2}{a,b,c}2^a\|\beta\|_2^{2b} (\sigma^2-1)^c \langle \beta^{\otimes a}\otimes e_y^{\otimes n-a}, H_n(X)\rangle  \\
& = \frac{(n-1)!!}{\sqrt{n!}} \sum_{a+b+c=n/2} \binom{n/2}{a,b,c}2^a\ (\sigma^2-1)^c \langle M_{a+2b}\otimes e_y^{\otimes n-a}, H_n(X)\otimes I_d^{\otimes b}\rangle \\
& = \frac{(n-1)!!}{\sqrt{n!}} \sum_{a+b+c=n/2} \binom{n/2}{a,b,c}2^a\ (\sigma^2-1)^c \E_Y\left[\langle M_{a+2b}\otimes e_y^{\otimes n-a}, H_n(X,Y)\otimes I_d^{\otimes b}\rangle\right] \;.
\end{align*}
Using Proposition \ref{prop:main} and Lemmas \ref{MLRMomentCompLem} and \ref{lem:cov-hxy},
we can approximate $\langle M_{a+2b}\otimes e_y^{\otimes n-a}, H_n(X,Y)\otimes I_d^{\otimes b}\rangle$
to expected squared error $O(1)^n/\sqrt{N}$ with $N$ samples in $\poly(N,k,d)$ time.
As the above sum has $\poly(n)$ terms with factors $O(1)^n$ (so long as $\sigma\leq 1$),
we can approximate $\langle T_n, H_n(X)\rangle$ for random Gaussian $X$
to expected error $O(1)^n/\sqrt{N}$ in polynomial time.

Unfortunately, we cannot compute infinitely many values of $\langle T_n,H_n(x)\rangle$,
so we will want to show that we can truncate the sum.
This amounts to showing that $\|T_n\|_2$
is small when $n$ is sufficiently large. For this, we note that
$$
T_n = \frac{(2n-1)!!}{\sqrt{n!}} \sum_{i=1}^k w_i  \mathrm{Sym}\left( (2 \beta_i \otimes e_y + (\|\beta_i\|_2^2 + \sigma^2 -1) e_y \otimes e_y)^{\otimes n/2} \right) \;.
$$
First, using the triangle inequality, we note that
$$
\|T_n\|_2 \leq \max_i \left\| \mathrm{Sym}\left( (2 \beta_i \otimes e_y + (\|\beta_i\|_2^2 + \sigma^2 -1) e_y \otimes e_y)^{\otimes n/2} \right)\right\|_2 \;.
$$
Next we note that the matrix $(2 \beta_i \otimes e_y + (\|\beta_i\|_2^2 + \sigma^2 -1) e_y \otimes e_y)$
is similar to some $2\times 2$ matrix
$M_i = \left[\begin{matrix} \gamma_{1,i} & 0 \\ 0 & \gamma_{2,i} \end{matrix} \right].$
Thus, we have that
\begin{align*}
\|T_n\|_2^2 & \leq \max_i \left\| \mathrm{Sym}(M_i^{\otimes n/2}) \right\|_2^2\\
& = \max_i\left\|\sum_{a=0}^{n/2} \binom{n/2}{a} \gamma_{1,i}^a \gamma_{2,i}^{n/2-a} \mathrm{Sym}(e_1^{\otimes 2a} \otimes e_2^{\otimes n-2a})\right\|_2^2\\
& = \max_i\sum_{a=0}^{n/2} \binom{n/2}{a}^2 \gamma_{1,i}^{2a} \gamma_{2,i}^{n-2a} \binom{n}{2a}^{-1} \;,
\end{align*}
where the last line is because symmetrizations of
$e_1^{\otimes 2a} \otimes e_2^{\otimes n-2a}$ for different values of $a$ are orthogonal,
and two permutations of $e_1^{\otimes 2a} \otimes e_2^{\otimes n-2a}$
are orthogonal unless the copies of $e_1$ end up in the same places
(in which case they have dot product $1$). Thus, we conclude that
$$
\|T_n\|_2^2 \leq n \max_{i,a} \gamma_{1,i}^{2a} \gamma_{2,i}^{n-2a} \leq n \gamma^n \;,
$$
where $\gamma$ is the largest absolute value of any eigenvalue
of any of the matrices $(2 \beta_i \otimes e_y + (\|\beta_i\|_2^2 + \sigma^2 -1) e_y \otimes e_y)$.

To bound the latter quantity, note that $M_i$ has determinant $-\|\beta_i\|_2^2$
and trace $\|\beta_i\|_2^2 + \sigma^2 -1$. Therefore, it has eigenvalues
$$
\frac{\|\beta_i\|_2^2 + \sigma^2 - 1 \pm \sqrt{(\|\beta_i\|_2^2 + \sigma^2 - 1)^2+4\|\beta_i\|_2^2}}{2} \;.
$$
Since $(\|\beta_i\|_2^2 + \sigma^2 - 1)^2+4\|\beta_i\|_2^2 \leq (1+\|\beta_i\|_2^2)^2$,
these eigenvalues are between
$$
\frac{\|\beta_i\|_2^2 + \sigma^2 - 1 \pm (1+\|\beta_i\|_2^2)}{2}
$$
or between $-1+\sigma^2/2$ and $\|\beta_i\|_2^2+\sigma^2/2.$
In particular, if $B,\sigma \leq 1/2$
we have that $\gamma \leq 1-\sigma^2/2$.

Therefore, if this holds, up to $L_2$ error $\eps/2$, we have that
$$
(F/G)(x) = \sum_{n=0}^{O(\sigma^{-2} \log(1/\eps))}\langle T_n, H_n(x)\rangle,
$$
which can be approximated to error $\eps/2$
using $N= \poly(k,d)\eps^{O(\sigma^{-2})} $ samples and $\poly(N, d)$ time.

Finally, noting that $F$ is a mixture of Gaussians with covariances bounded
below by $\Omega(\sigma^2)$ and above by a constant, it is not hard to see
that except with $\eps$ probability, a point sampled from $F$ has $F/G(x)$
at most $\poly(1/(\eps \sigma))$. Thus, if $C$ is a big enough polynomial in $1/(\eps\sigma)$,
then picking a random sample $x\sim G$ and accepting with probability
$\min(1,f(x)/C)$, where $f(x)$ is our approximation to $(F/G)(x)$,
gives an $\eps$-approximation to the distribution $F$.

This concludes the proof of Theorem~\ref{MLRThm}. \qed

\newpage

\section{Conclusions} \label{sec:concl}

In this work, we gave a general efficient algorithm to approximate higher-order
moment tensors, as long as there are reasonable unbiased estimators
for these moments that can be efficiently sampled.
This type of tensors arise in a range of learning problems.
We leveraged our general result to obtain the first polynomial-time
algorithms for learning mixtures of linear regressions, learning sums of ReLUs and cosine activations,
density estimation for mixtures of spherical Gaussians with bounded means,
and parameter estimation for mixtures of spherical Gaussians under optimal separation.
In all cases, our learning algorithms run in $\poly(d, k)$ time.

A number of open questions remain. The most obvious ones are quantitative:
Specifically, can the dependence on $1/\eps$ for learning sums of ReLUs be improved to polynomial?
Can we remove the bounded means assumption for density estimation of mixtures of spherical
Gaussians? Can we obtain a polynomial-time parameter estimation algorithm for mixtures of linear regressions?
More broadly, for what other learning tasks is the implicit moment tensor
technique applicable? These are left as interesting directions for future work.

\bibliographystyle{alpha}
\bibliography{allrefs}

\newpage

\end{document}